\documentclass[sigplan, 10pt]{acmart}

\usepackage{Style-Files/preamble}

\newcommand{\nc}[1]{{\color{cyan}\grumbler{Natacha}{#1}}}
\newcommand{\neil}[1]{{\color{blue}\grumbler{Neil}{#1}}}
\newcommand{\fs}[1]{{\color{purple}\grumbler{Florian}{#1}}}

\newcommand{\la}[1]{{\color{olive}\grumbler{LA}{#1}}}

\acmYear{2024}\copyrightyear{2024}
\acmConference[SOSP '24]{ACM SIGOPS 30th Symposium on Operating Systems Principles}{November 4--6, 2024}{Austin, TX, USA}
\acmBooktitle{ACM SIGOPS 30th Symposium on Operating Systems Principles (SOSP '24), November 4--6, 2024, Austin, TX, USA}
\acmDOI{10.1145/3694715.3695942}
\acmISBN{979-8-4007-1251-7/24/11}

\usepackage{draftwatermark}

\SetWatermarkText{{\ifnum\value{page}=16 The following material has not been Peer-Reviewed \else\fi}}
\SetWatermarkColor[gray]{0.5}
\SetWatermarkFontSize{0.4cm}
\SetWatermarkScale{1}
\SetWatermarkLightness{0.6}
\SetWatermarkAngle{0}
\SetWatermarkVerCenter{2.3cm}
\SetWatermarkHorCenter{10.8cm}

\title{\sys: Seamless high speed BFT}
\author{Neil Giridharan}
\authornote{Joint first authors, alphabetical.}
\affiliation{\institution{UC Berkeley}\country{}}
\author{Florian Suri-Payer}
\affiliation{\institution{Cornell University}\country{}}
\authornotemark[1]
\author{Ittai Abraham}
\affiliation{\institution{Intel Labs}\country{}}
\author{Lorenzo Alvisi}
\affiliation{\institution{Cornell University}\country{}}
\author{Natacha Crooks}
\affiliation{\institution{UC Berkeley}\country{}}



\begin{abstract} 
Today's practical, high performance Byzantine Fault Tolerant (BFT) consensus protocols operate in the partial synchrony model. 
However, existing protocols are inefficient when deployments are indeed \textit{partially} synchronous. 
They deliver either low latency during fault-free, synchronous periods (\textit{good intervals}) or robust recovery from events that interrupt progress (\textit{blips}). At one end, traditional, view-based BFT protocols optimize for latency during good intervals, but, when  blips occur, can suffer from performance degradation (\textit{hangovers}) that can last beyond the return of a good interval.
At the other end, modern DAG-based BFT protocols recover more gracefully from blips, but exhibit lackluster latency during good intervals.
To close the gap, this work presents \sys{}, a novel high-throughput BFT protocol that offers both low latency and \textit{seamless} recovery from blips. By combining a highly parallel asynchronous data dissemination layer with a low-latency, partially synchronous consensus mechanism, \sys{} \one avoids the hangovers incurred by traditional BFT protocols and \two matches the throughput of state of the art DAG-based BFT protocols while cutting their latency in half, matching the latency of traditional BFT protocols.

\end{abstract}

\begin{document}

\maketitle
\pagestyle{empty} 


\section{Introduction}

This work presents \sys, a Byzantine Fault Tolerant (BFT) state machine replication (SMR) protocol that sidesteps current tradeoffs between low latency, high throughput, and robustness to faults.


BFT SMR offers the appealing abstraction of a centralized, trusted, and always available server, even as some replicas misbehave. Consensus protocols implementing this abstraction~\cite{castro1999pbft,yin2019hotstuff,Bullshark} are at the core of recent interest in decentralized systems such as blockchains, central-bank endorsed digital currencies~\cite{digital-euro, digital-euro-2}, and distributed trust systems. CCF, a Trusted Execution Environment (TEE) enabled BFT system~\cite{microsoft-ccf}, for instance, is used to power Azure SQL's integrity functionality~\cite{ccf-ledger}. 
These protocols must offer
\one high throughput, \two low end-to-end latency, and \three robustness to both faults and changing network conditions; all whilst being simple enough to build and deploy. 



Today's popular (and actually deployed) BFT consensus protocols operate primarily in the {\em partial synchrony model}~\cite{dwork1988consensus}, introduced to get around the impossibility of a safe and live solution to consensus in an asynchronous system~\cite{fischer1985impossibility}. Partially synchronous protocols are always safe, but provide no liveness or performance guarantees before some ``magic'' \textit{Global Stabilization Time} (GST), after which synchrony (all messages are received within a timebound $\Delta$) is expected to hold forever on; they optimize for latency \textit{after} GST holds. 
Deployed systems simulate this model using exponentially increasing timeouts -- pragmatically, they declare GST attained and $\Delta$ found once timeouts are  large enough that they are not being violated. In this regime, user requests (transactions) are guaranteed to commit.\fs{what about faults? I think the sentence reads fine without (and is correct)}



In practice, however, this happy regime only holds intermittently. GST is an elegant fiction, as real deployments are subject to unpredictable replica and network failures, network adversaries (DDoS/route hijacking), latency fluctuations, link asymmetries, etc. that can cause periods in which progress stalls ({\em blips}). These begin when a timeout is violated and end once timeouts are again met (perhaps by adjusting the timeout length, or removing a presumed faulty leader).
Unlike theoretical partial synchrony, which promises endless bliss after an initial (asynchronous) storm, {\em real} partial synchrony is instead piece-wise: periods where timeouts are met are separated by blips, during which progress is stalled.

Unfortunately, we find that existing partially synchronous protocols are not actually efficient under real partial synchrony: they have to choose between achieving low-latency when timeouts are met or remaining robust to blips.

\par \textbf{Traditional BFT} \cite{castro1999pbft, yin2019hotstuff, kotla10zyzzyva, gueta2019sbft} protocols and their multi-leader variants \cite{stathakopoulou2019mir, stathakopoulou2022state, gupta2021rcc} optimize for performance after GST ({\em i.e.,} in the absence of timeout events). This design choice  allows these protocols to minimize message exchanges in the common (synchronous) case. Unfortunately, we find that, after even a brief blip, these protocols cannot guarantee a graceful resumption of operations once synchrony returns; instead, they experience what we characterize as a {\em hangover}. Specifically, under sustained high load, the lack of progress caused by a blip (and the consequent loss of throughput) can generate large request backlogs that cause a degradation in end-to-end latency that persists well after the consensus blip has ended (Figs. \ref{fig:hangover},\ref{fig:leader-fail},\ref{fig:partition}).
Thus, while on paper traditional BFT protocols enjoy low consensus latency, their poor resilience to blips can cause them to underperform under real-life partial synchrony.   
\vspace{2.5mm}
\par \textbf{Directed Acyclic Graph (DAG) based BFT,} a newly popular class of consensus protocols \cite{baird2016swirlds, keidar2021all, danezis2022narwhal, malkhi2022maximal, Bullshark, spiegelman2023shoal} takes an entirely different approach, steeped in an \textit{asynchronous} system model. Asynchronous protocols optimize for worst-case message arrivals, and leverage randomness to guarantee progress without relying on timeouts. 
DAG systems employ as backbone a high throughput asynchronous data dissemination layer which, through a series of structured, tightly synchronized rounds, forms a DAG of temporally related data proposals. Replicas eventually converge on the same DAG and deterministically \textit{interpret} their local view  to establish a consistent total order. 
This approach yields excellent throughput and reduces the effects of blips on the system. Unfortunately, these benefits come at the cost of prohibitive latency; consequently, these protocols see little practical use. In pursuit of lower latency, recent (deployed) DAG protocols \cite{spiegelman2022bullshark, spiegelman2023shoal} switch to  partial synchrony and reintroduce timeouts,  but still require several rounds of Reliable Broadcast~\cite{bracha1985asynchronous} communication. For instance, Bullshark~\cite{Bullshark}, \changebars{the most popular partially synchronous DAG protocol}{the state of the art partially synchronous DAG protocol at the core of Sui~\cite{mysten} and Aptos~\cite{aptos}}, still requires up to 12 message delays (mds) to commit transactions. 

\vspace{2.5mm}

This paper shows that, thankfully, these tradeoffs are not fundamental. We propose \sys{}, a new consensus protocol that minimizes hangovers without sacrificing low latency during synchronous periods. 
\sys's architecture is inspired by recent DAG protocols: it constructs a highly parallel data dissemination layer that continues to make progress at the pace of the network ({\em i.e.,} at the lowest rate of transfer between correct nodes) even during blips.
Atop, it carefully layers a traditional-style partially synchronous consensus mechanism that orders the stream of data proposals by committing concise state cuts. Like any partially synchronous protocol, consensus may fail to make progress during blips -- however, upon return of progress, \sys{} can instantaneously commit the \textit{entire} data backlog (with complexity independent from its size), minimizing the effect of hangovers.

\fs{In comments are some more differences to DAGs.}


\vspace{2.5mm}
Our results are encouraging (\S\ref{sec:eval}): \sys{} matches the throughput of Bullshark~\cite{Bullshark} while reducing its latency by more than half; and it matches the latency of HotStuff~\cite{yin2019hotstuff} without suffering the hangovers that HotStuff experiences. \fs{and hopefully we can show it also beats out Bullshark given partitions!}
\vspace{1cm}

This paper makes three core contributions:

\begin{itemize}
        \item It formalizes two new notions -- {\em hangovers} and \textit{seamless} partial synchrony (\S\ref{sec:psync}) -- to characterize the performance of partially synchronous protocols in practice.
    \item It revisits the set of system properties required to design a BFT consensus protocol that avoids protocol-induced hangovers ({\em i.e.} it is seamless).
    \item It presents the design of \sys{}, a novel \textit{seamless} partially synchronous BFT consensus protocol that offers high throughput and low latency.
\end{itemize}

\section{Partial Synchrony: Theory meets Practice}
\label{sec:psync}

The partial synchrony model is an appealing theoretical framework for reasoning about liveness in spite of the seminal FLP impossibility result~\cite{fischer1985impossibility}.  Real systems can, in principle, approximate this model using timeouts: they exponentially increase the value of the timeout until they enter a regime where timeout violations subside, and progress becomes possible, as if the network had indeed reached its Global Stabilization Time.
Unfortunately, timeouts are notoriously hard to set~\cite{tennage2023quepaxa, spiegelman2023shoal, liu2023flexible}):
too low of a timeout, and even common network fluctuations stop progress; too high, and the protocol is slow to respond to faults.

In practice, timeout violations are simply inevitable. Wide area network (WAN) latencies
are unpredictable, can vary by order of magnitudes~\cite{hoiland2016measuring}, and many RPCs have been shown to suffer
from high tail latency~\cite{dean2013tail, seemakhupt2023cloud}.
The operating context for today's blockchain systems is even more challenging.  Participants in these systems can
have highly asymmetric network and hardware resources, and widely different deployment experiences, all factors that can greatly impact how well they can keep up with requests. These asymmetries complicate
timeout configurations and result in frequent timeout events~\cite{priv-com-aptos-mysten}.

More fundamentally, accounting for the possibility of repeated timeout
violations is simply part of the job description for any system that
aims to be resilient to Byzantine faults.  Malicious replicas may, for
example, intentionally fail to send a required message, or a network
adversary may target correct participants via DDoS/route hijacking attacks. Indeed, recent work, for instance, demonstrates that targeted
network attacks on leaders are surprisingly easy to mount and drastically hurt performance~\cite{consensus-dos}.

Whatever the cause may be, the corrective actions that BFT SMR protocols take in response to a timeout violation ({\em e.g.,}  electing a new consensus leader) can lead to {\em blips}, {\em i.e.,} periods
during which progress stalls.
Practical deployments are thus truly \textit{partially}
synchronous. Messages will not trigger timeouts for a while, 
but these periods of synchrony will be regularly interrupted by
blips.  To properly assess the performance of a partially synchronous
system, we should consider how it behaves \textit{at all times},
during periods of synchrony, but also after recovering from 
blips.  To this effect, we introduce the dual notions of \textit{hangovers} and \textit{seamlessness}. Informally, hangovers consist of performance degradations triggered by a blip that persist once the blip is over; we say that a  partially synchronous protocol is seamless if it does not suffer from {\em protocol-induced} hangovers (defined in \S\ref{sec:hangovers-and-seamlessness} below).  
\subsection{Hangovers \& Seamlessness}
\label{sec:hangovers-and-seamlessness}
Following Aardvark~\cite{aardvark} and Abraham et al.~\cite{good-case} we say that an interval is \textit{good} if the system is synchronous (w.r.t. some implementation-dependent timeout on message delay), and the consensus process is led by a correct replica.
Intuitively, good intervals capture the periods during which progress is guaranteed; all non-good intervals are blips.
Additionally, we consider an interval to be {\em gracious}~\cite{aardvark} if it is good, and \textit{all} replicas are correct. This distinction is helpful to reason about specific optimizations, such as so-called {\em fast paths} in the protocol.
Given this, we define:
\begin{definition}
    A \textbf{hangover} is any performance degradation caused by a blip that persists beyond the return of a good interval. 
\end{definition}
Hangovers impact both throughput and latency. We characterize a throughput hangover as the duration required to commit all outstanding transactions
issued before the start of the next good interval. If the provided load is below the system's steady-state throughput, the system will eventually recover all throughput ``lost'' during the blip.\footnote{If the load is equal or greater to the throughput, then the lost committed throughput (or \textit{goodput}) can never be recovered.}
Delayed throughput recovery in turn increases latency, even for transactions issued \textit{during} a good interval that follows a blip. Before these transactions can commit, they must first wait for the backlog formed during the blip to commit. Hangovers can not only cause immediate latency spikes, but also make the system more susceptible to future blips. Good intervals do not last forever, and unresolved hangovers can add up---and continue indefinitely.

In contrast, in a hangover-free system, all submitted transactions whose progress is interrupted by a blip, commit instantaneously when progress resumes; further, these are the only transactions that experience a hike in latency, and for no more than the remaining duration of the blip.

Some hangovers are unavoidable, {\em e.g.,} those due to insufficient network bandwidth or to message delays. No consensus protocol, not even asynchronous ones, can provide progress beyond the pace of the network: 
without data, there isn’t anything to agree on! Consider, for instance, the case of a network partition that causes ten thousand 1MB transactions to accumulate; it will take at least 80 seconds, on a 1Gbit link, for them to be propagated once the partition resolves. \fs{and: new tx that are being submitted after the end of partition will also become delayed}

Other hangovers, however, are the result of suboptimal system design, where protocol logic (timeouts, commit rule, etc) introduces unnecessary delays. We refer to these hangovers are \textit{protocol-induced}. As we describe further in \S\ref{sec:existing-approaches}, the most common type of protocol-induced hangovers for consensus is the artificial coupling of data dissemination with the protocol's ordering logic.  

Thankfully, protocol-induced hangover are not inevitable and can be side-stepped through careful protocol design. Perhaps obviously, a protocol should accomplish this \textit{without} making itself more susceptible to blips; a protocol that introduces \textit{more blips} may implicitly introduce hangovers, or lose progress entirely. 

We say that a consensus protocol is \textit{seamless} if it behaves optimally in the face of blips. Formally:
\begin{definition} 
   A partially synchronous system is \textbf{seamless} if (i) it experiences no protocol-induced hangovers, and (ii) it does not introduce any mechanisms\footnote{Beyond the timeouts that are, of course, necessary for liveness.} that make the protocol newly susceptible to blips.
\end{definition}

Seamlessness is a desirable property for any system, but especially relevant to BFT systems. A Byzantine leader can easily cause a blip: it can simply fail to propose. If this blip causes a hangover, it will continue to impact the performance of the system even after the leader has been replaced. This is clearly undesirable, especially because, after all, BFT protocols should be robust to Byzantine faults~\cite{aardvark}.

\subsection{Existing Approaches}\label{sec:existing-approaches}
Traditional BFT protocols~\cite{castro1999pbft, yin2019hotstuff} are not seamless. 
They tightly couple data dissemination and ordering: for each batch of transactions on which it seeks to achieve consensus,  the protocol's current leader jointly disseminates the data associated with the batch and proposes how this batch should be ordered relative to other batches.
Consensus is then reached on each batch independently. If the consensus logic stalls, {\em e.g.,} because of a Byzantine leader triggering a timeout, so does data dissemination.  When a good interval resumes, the time necessary to commit ({\em i.e.,} transmit, verify and order) the backlog of transactions accumulated during the blip will thus be on the order of the backlog's size.

Take for instance, the behavior of HotStuff (HS) ~\cite{yin2019hotstuff}, a state of the art BFT protocol at the core of the former Diem blockchain~\cite{baudet2019state}. We use a 1s (doubling) timeout (the default in many companies~\cite{priv-com-aptos-mysten}), and trigger a three second blip by simulating a single leader failure; because HS pipelines rounds and eagerly rotates leaders, a single failure can cause two timeouts to trigger (\S\ref{sec:eval}).
After the blip ends, HS, burdened with individually disseminating and processing the transaction batches accumulated in its backlog, suffers from a hangover that is 30\% longer than the original blip itself! (Figure~\ref{fig:hangover}).
\begin{figure}[h!]
    \vskip -2pt
   \centering
   \includegraphics[width=0.45\textwidth]{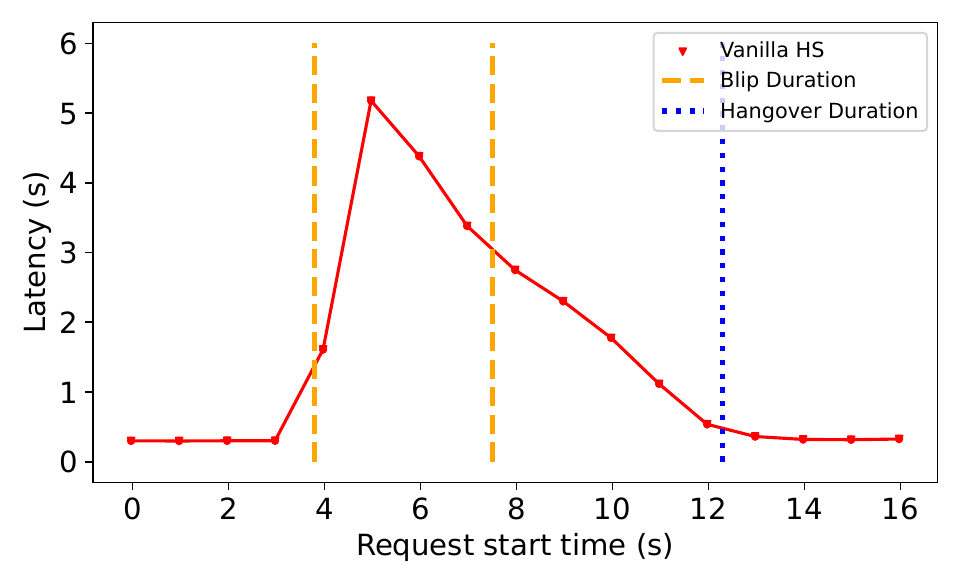}
   \vskip -2pt
   \caption{Latency Hangover in HotStuff~\cite{yin2019hotstuff}.}
   \label{fig:hangover}
   \vskip -2pt
\end{figure}

Simply proposing larger batches is not a viable option for reducing the number of \changebars{batches in the backlog}{required agreement instances}, as it comes at the cost of increased latency in the common case.
\fs{suggested re-write starting here (in changebars below)}
A naive decoupling of data dissemination and consensus ordering logic (by proposing only a batch digest)~\cite{castro1999pbft, bagaria2019prism, danezis2022narwhal} can help scale throughput, but is not robust: \one malicious proposers may overwhelm correct replicas with data, and \two replicas that are out of sync must fetch missing data (data synchronization) before voting, which is on the timeout-critical path of consensus. We demonstrate empirically in \S\ref{sec:eval} that such synchronization effects are common even in the absence of blips, and restrict throughput scalability. 


Recent partially-synchronous DAG-based protocols~\cite{danezis2022narwhal, Bullshark}, in contrast, come {\em close} to seamlessness. 
These protocols decouple data dissemination from the consensus ordering logic, and construct an asynchronously growing DAG consisting of batches of transactions that are chained together. The construction leverages Reliable Broadcast (RB) to generate proofs that transitively vouch for the availability of those batches and of all topologically preceding batches in the DAG. The consensus ordering logic can then commit entire backlogs with constant cost by proposing only a single batch. Commit complexity is thus \textit{independent} of blip duration. 

Current DAG protocols allow data dissemination to proceed at the pace of the network, rather than that of consensus, but do so at the price of slowing down consensus, and thus fall short of seamlessness. In particular, they still place data synchronization on the timeout-critical path (before voting), which may introduce protocol-generated blips; further, they need to recursively traverse the DAG backwards to infer all topologically preceding data. This may require fetching data for multiple sequential rounds, delaying commit.\footnote{In practice, DAGs minimize worst-case synchronization by forcing (a supermajority of) replicas to advance through rounds in lock-step.} 
\fs{NOTE: Because DAGs move kind of in lock-step (they "slow" down how fast people can advance) they are usually fairly in-sync, and the above effects won't be as frequent (as for the naive BatchedHS); but at a principled level they remain.}

Besides, existing DAG protocols  exhibit during {\em good} intervals up to twice the latency of traditional BFT protocols. Consensus safety is hardwired into the structure of the DAG, and establishing it requires up to four rounds of RB (each comprising three message delays) -- for a total of 12 message delays. 

This paper asks: is it possible to develop a low-latency and high-throughput \textit{seamless} protocol? We answer this question in the positive with \sys{}, a novel BFT protocol. Autobahn instantly commits all backlogged proposals upon returning from a blip\changebars{, and without making consensus more susceptible to blips}{} (thus achieving seamlessness).
\nc{I'm nto sure we need the susceptible part, just for this question}
\section{\sys{}: Overview}
\label{sec:overview}
 \sys{}  comprises two logically distinct layers: a horizontally scalable data dissemination layer that always makes progress at the pace of the network\changebars{}{, even during blips}; and a low-latency,  partially synchronous consensus layer that tries to reach agreement on snapshots of the data layer. Decoupling consensus from data dissemination is not unique to \sys{}~\cite{castro1999pbft, bagaria2019prism, danezis2022narwhal}; its 
 core innovation lies in more cleanly separating the responsibility of each layer while carefully orchestrating their interaction. Seamlessness requires two key properties:
\begin{itemize} 
\item \textit{Responsive transaction dissemination.} The data dissemination layer should proceed at the pace of the network. We say that a transaction has been successfully disseminated once it is delivered by one correct replica.
 \item \textit{Streamlined Commit.} During good intervals, all successfully disseminated proposals should be committed in bounded time  and with cost and latency independent of the number of transactions to commit. 
\end{itemize} 
\fs{I'm wondering if this paragraph is even necessary? It basically re-states the above?}\nc{Removed}

Implementing a responsive data dissemination layer may appear  straightforward: simply broadcast all transactions to all replicas in parallel. This naive strategy unfortunately makes it extremely challenging to implement a consensus layer that streamlines commit. \fs{maybe the next two sentences can be cut?} The consensus layer would have to individually propose (references to) all disseminated transactions, which would either require replicas to fetch and verify locally missing transactions (data synchronization) prior to voting in consensus, or risk breaking liveness in the presence of Byzantine proposers.
Achieving seamlessness and low latency instead requires injecting the right amount of \changebars{coordination}{structure} in the data layer's design: too little and the consensus layer cannot be designed to streamline commit; too much, and the good-case latency may spike. 
We find that the dissemination layer must specifically support \textit{instant referencing}: it must allow the consensus layer to uniquely identify, and prove the availability of, the set of disseminated transactions. This must be done in constant time and space (w.r.t. to blip length), and without necessarily being in possession of all data. Note that the bounded time requirement implicitly requires that there be no data synchronization  on the critical (voting) path of consensus, as this could lead to timeout violations, resulting in further blips (\textit{non-blocking sync}). Moreover, any data syncing should finish by the time consensus commits \changebars{to avoid delays that may cause hangovers (\textit{timely sync}).}{(\textit{timely sync})--in practice, in a single round-trip--to avoid delays that may cause hangovers. }

These principles shape \sys{}'s data dissemination and consensus layers.
\par \textbf{Data dissemination.} 
All replicas in \sys{} act as proposers, responsible for disseminating transactions received by clients. Each replica, in parallel, broadcasts batches of transactions (data proposals) and constructs a local chain (a \textit{data lane}) that implicitly assigns an ordering to all of its data proposals. This structure allows us to transitively prove the availability of all data proposals in a chain in constant time, thus guaranteeing \textit{instant referencing}. Data lanes grow independently of one another, subject only to load and resource capacity, and are agnostic to blips (guaranteeing responsiveness). 
Finally, \sys guarantees that all successfully disseminated data proposals will commit during good intervals (\textit{reliable inclusion}); this holds true even for Byzantine proposers, and ensures that replicas cannot "wastefully" disseminate transactions without intending to commit them.

\par \textbf{Consensus.} The consensus layer exploits the structure of data lanes to, in a single shot, commit an arbitrarily large data lane state. \sys efficiently summarizes data lane state as a cut containing only each replica's latest proposal (its \textit{tip}), using instant referencing to uniquely identify a lane's history. 
Consensus itself follows a classical (latency optimal) PBFT-style coordination pattern~\cite{castro1999pbft, jolteon-ditto} consisting of two round-trips; during gracious intervals, \sys{}'s \textit{fast path} commits in only a single round-trip.
\fs{brought back with some edits}
After committing to a tip cut, a replica can infer the respective chain-histories of each tip and synchronize on any locally missing data (some transactions in the cut may not have been replicated to all nodes). 
To ensure that data synchronization doesn't stall consensus progress, \sys{} carefully constructs voting rules in the data dissemination layer to allow synchronization to complete in parallel with agreement, and in a single round-trip.

Once a replica is in possession of the data included in each tip's history data, it deterministically interleaves the $n$ data lanes to construct a single total order.

This design allows a single consensus proposal to \one reach throughput that scales in the number of replicas, \two reach agreement with cost independent of the size of the data lanes, and \three preserve the latency-optimality of traditional BFT protocols. 
We describe the full \sys{} protocol in \S\ref{sec:protocol}.
\section{Model}
\label{sec:model}

\sys{} adopts the standard assumptions of prior BFT work \cite{castro1999pbft, yin2019hotstuff, Bullshark}, positing $n = 3f+1$ replicas, at most $f$ of which are faulty.
We consider a participant (client or replica) correct if it adheres to the protocol specification; a participant that deviates is considered faulty (or Byzantine).  We make no assumptions about the number of faulty clients.
We assume the existence of a strong, yet static adversary, that can corrupt and coordinate all faulty participants' actions; it cannot however, break standard cryptographic primitives. 
Participants communicate through reliable, authenticated, point-to-point channels. We use $\langle m \rangle_r$ to denote a message $m$ signed by replica $r$. We expect that all signatures and quorums are validated, and that external data validity predicates are enforced; we omit explicit mention of this in protocol descriptions.

\sys{} operates under partial synchrony~\cite{dwork1988consensus}: it makes no synchrony assumptions for safety, but guarantees liveness only during periods of synchrony~\cite{fischer1985impossibility}. \fs{should we say that we assume no known Delta for our system? It is purely a proof technique, and equivalent. "There exists an upper bound delta but is NOT known to the system."}\nc{I think its ok}\la{I think so too.}

\section{\sys{}: The Protocol}
\label{sec:protocol}

To describe \sys{} more thoroughly, we first focus on the data dissemination layer (\S\ref{sec:datalayer}), which guarantees responsive data dissemination, even during blips.  We then focus on how the consensus layer (\S\ref{sec:core_consensus}) leverages the data \changebars{layer's}{layers} lane structure to efficiently commit arbitrarily long data lanes with cost and latency independent of their length, thus seamlessly recovering from blips.

\subsection{Data dissemination}
\label{sec:datalayer}

\begin{figure*}[h!]
\centering
\begin{minipage}{.5\textwidth}
  \centering
  \includegraphics[width=\linewidth]{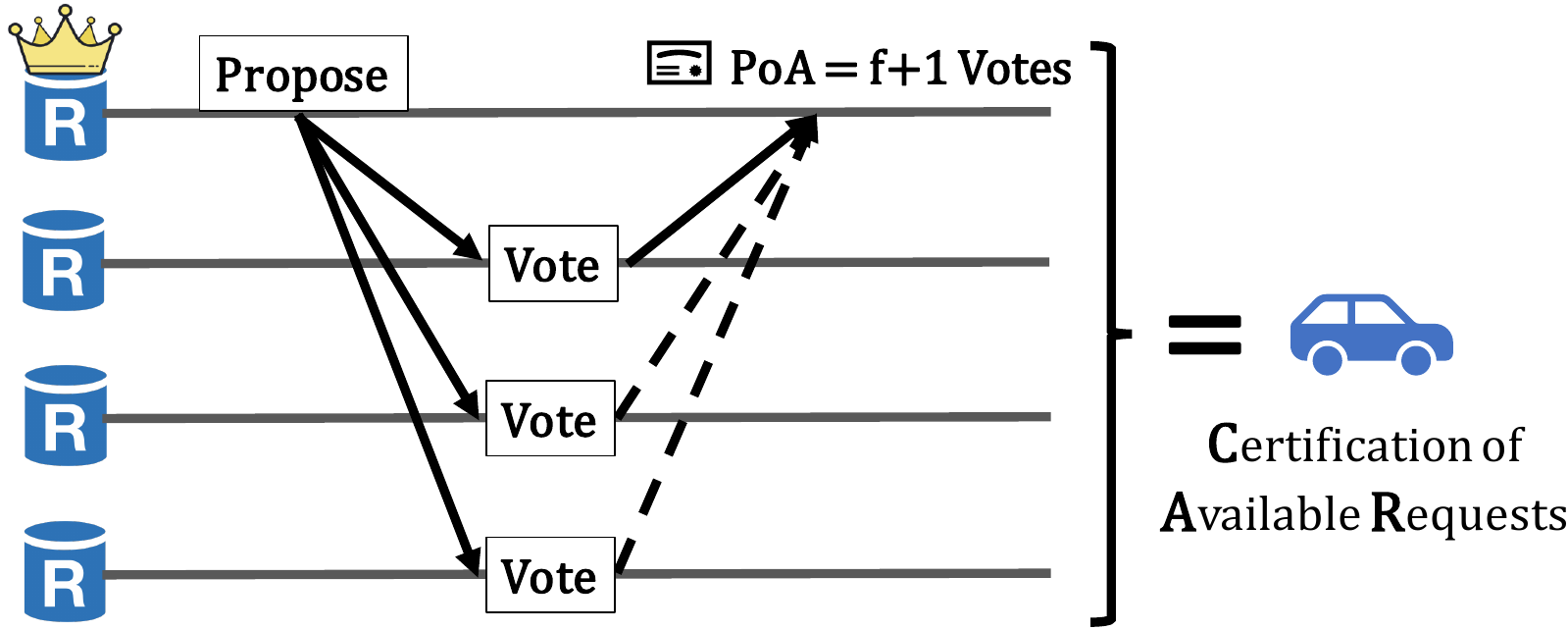}
  \vskip 8pt
  \caption{Car protocol pattern.}
  \label{fig:car}
\end{minipage}%
\begin{minipage}{.5\textwidth}
  \centering
  \includegraphics[width=.8\linewidth]{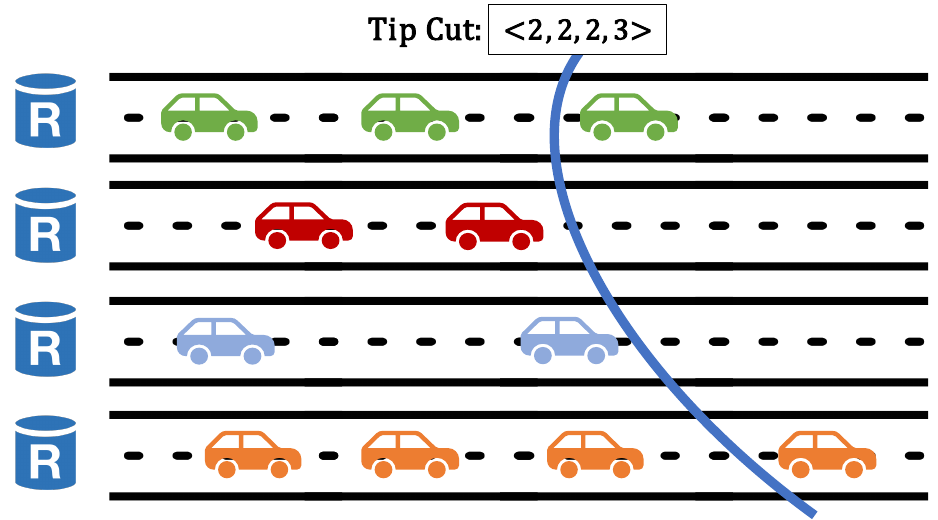}
  \vskip 0pt
  \caption{\sys{} parallel data lanes; example cut.}
  \label{fig:car_lanes}
\end{minipage}
\vskip -5pt
\end{figure*}


\sys{}'s data dissemination layer offers the necessary structure for the consensus layer to ensure a streamlined commit. It precisely provides to consensus the three properties of instant referencing, non-blocking sync and timely-sync. Crucially, it provides no stronger guarantee, as additional properties may cause tail latency to increase. 
\par \textbf{Lanes and Cars.} In \sys{}, every replica acts as a proposer, continuously batching and disseminating incoming transactions. Each replica proposes new batches of these transactions (data proposals) at its own rate, and fully independently of other replicas — we say that each replica operates in its own \textit{lane}, and as fast as it can. 

Within its lane, each replica leverages a simple Propose and Vote message pattern: the proposer broadcasts data proposals to all other nodes, and replicas vote to acknowledge delivery. $f+1$ vote replies constitute a Proof of Availability (\textsc{PoA}) that guarantees  at least one correct replica is in possession of the data proposal and can forward it to others if necessary.
We call this simple protocol pattern a \textit{car} (Certification of Available Request), illustrated in Figure \ref{fig:car}.
Note that \sys{}, unlike most DAG-BFT protocols, does not force all data proposals to go through a round of \textit{reliable broadcast}~\cite{danezis2022narwhal,Bullshark}, which must process $n-f$ votes in order to achieve non-equivocation. Reliable, non-equivocating broadcast is not necessary in \sys{} to achieve the desired data layer properties  (\S\ref{sec:consensus_on_lanes}).

A lane is made up of a series of cars that are chained together (Fig. \ref{fig:car_lanes}). A proposer, when proposing a new data proposal, must include a reference to its previous car's data proposal. Similarly, replicas will vote for a car at position $i$ if and only if its proposal references a proposal for car $i-1$, which the replica has already received and voted on. 
This construction ensures that a successful car for block $i$ transitively proves the availability of a car for all blocks $0$ to $i-1$ in this proposer's lane. Validating the head of the lane is thus sufficient to reference arbitrary lane state (enabling instant referencing) and confirm that it has been disseminated (necessary for non-blocking sync). Requiring replicas to always vote \textit{in-order} further ensures that at least one correct replica is in possession of all proposals in this lane, allowing the consensus layer to (\S\ref{sec:synchronization}) to experience timely-sync.

\par \textbf{Protocol specification.} Each replica maintains a local view of all lanes. 
We call the latest proposal of a lane a tip; a tip is certified if its car has completed (a matching \textsc{PoA} exists). 

\vskip -4pt
\begin{algorithm}
\caption{Data layer cheat sheet.}\label{alg:data-state}
\begin{algorithmic}[1]
\State{$car$}~~~~~~~~$\triangleright$ \textit{Propose} \& \textit{Vote} pattern
\State{\textsc{PoA}}~~~~~~$\triangleright$ $f+1$ matching \textsc{Vote}s for proposal
\State{$pass$}~~~~~~$\triangleright$ passenger: (proposal, Option(\textsc{PoA}))
\State{$pos$}~~~~~~~~$\triangleright$ car position in a lane (proposal seq no)

\State{$lane$}~~~~~~~$\triangleright$ map: [pos $\rightarrow$ passenger] \label{var:lane}
\State{$lanes$}~~~~~$\triangleright$ map: [replica $\rightarrow$ lane] \label{var:lanes}
\State{$tip$}~~~~~~~~~$\triangleright$ last passenger in a lane ($pos = lane.length$) 
\end{algorithmic}
\end{algorithm}
\vskip -0pt

\par \protocol{\textbf{1: P} $\rightarrow$ \textbf{R}: Replica $P$ broadcasts a new data proposal \textsc{Prop}.}
$P$ assembles a batch $B$ of transactions, and creates a new data proposal for its lane $l$. 
It broadcasts a message $\textsc{Prop} \coloneqq (\langle pos, B, parent \rangle_{P}, cert)$ where $pos \coloneqq l.length+1$ (the position of the proposal in the lane), $parent \coloneqq h(l.tip.prop)$ (the hash of previous lane tip), and $cert \coloneqq l.tip.\textsc{PoA}$ (the availability proof of the previous proposal).

\par \protocol{\textbf{2: R} $\rightarrow$ \textbf{P}: Replica $R$ processes \textsc{Prop} and votes.}
$R$ checks whether the received \textsc{Prop} is valid: it checks that  \one it has already voted for the parent of this proposal  $(h({\mathit lanes}[P][pos-1].prop) == {\mathit parent})$ and \two it has not voted for this lane position before.
 $R$ then stores the parent's \textsc{PoA} received as part of the message as the latest certified tip $(\textit{lanes}[P][{\mathit pos-1}].\textsc{PoA} = {\mathit cert})$, and the current proposal as latest optimistic tip in $\mathit{lanes}[P][{\mathit pos}].{\mathit prop}$ $= \textsc{Prop}$. 
 Finally, it replies with a $\textsc{Vote} \coloneqq \langle dig=h(\textsc{Prop}), pos \rangle_R$.
If $R$ has not yet received the parent proposal\fs{ for \textsc{Prop}}, it buffers the request and waits. 

\par \protocol{\textbf{3: P} $\rightarrow$ \textbf{R}: Replica $P$ assembles \textsc{Vote}s and creates a \textsc{PoA}.}
The proposer aggregates $f+1$ distinct matching \textsc{Vote} messages  into a $\textsc{PoA} \coloneqq (dig, pos, \{\sigma_r\})$, which it will include in the next car; if no new batch is ready then $P$ can choose to broadcast the \textsc{PoA} immediately. 

Cars enforce only light (but sufficient) structure on the data layer. The number and timing of cars in each lane may differ; they depend solely on the lane owner's incoming request load and connectivity (bandwidth, point-latencies). 
Cars also do not preclude equivocation. While lanes governed by correct replicas will only ever consist of a single chain, Byzantine lanes may fork arbitrarily. This is by design: non-equivocation is not a property that needs to be enforced at the data layer and causes unnecessary increase in tail latency. 

\fs{TODO: One can pipeline/overlap cars to squeeze out even more latency -- this comes at a tradeoff with bounded wastage: waiting for the preceeding car's certificate guarantees there is only 1 optimistic/uncertified proposal at any time. If one sends them earlier, then there can be more. Can use the same pipelelining optimization from consensus: allow k uncertified cars to start (only start position p when p-k is certified. NOTE: this also affects optimistic tip optimisation: There might be up to k uncertified tips that need to be synced on before voting -- these can be synced in one go (1 RT) from the leader though: the leader may only propose tips for which it either has all parents, or a certificate. E.g. if it proposes p=3, and only p=1 has a cert, then it must have locally seen at least p=2 and p=3 so it can supply them. Basically, leader must ensure "local availability" property}

\subsection{From Lanes to Consensus}
\label{sec:consensus_on_lanes}
\fs{reads a bit redundant}
\nc{agreed, but I thought the redundancy was important to really hammer the properties down}
As stated, \sys{}'s data dissemination layer guarantees responsive data dissemination, instant referencing, and non-blocking, timely sync, but intentionally provides no consistency guarantee on the state observed by correct replicas. For instance, different replicas may observe inconsistent lane states (neither a prefix of another). 

The role of the consensus layer in \sys{} is then to reconcile these views and to totally order all certified data proposals. To ensure seamless recovery from blips, the commit process should be \textit{streamlined}: the consensus layer should be able to commit, in a single shot, an arbitrarily large number of certified data proposals. 
Unlike traditional BFT consensus protocols, that take as input a batch of transactions, a \textit{consensus proposal} in \sys{} consists of a snapshot \textit{cut} of all data lanes. 
Specifically, a consensus proposal contains a vector of $n$ \textit{certified} tip references, as shown in 
Figure \ref{fig:car_lanes}.

This design achieves two goals: \one committing a cut of all the lane states allows \sys to achieve consensus throughput that scales horizontally with the number of lanes ($n$ total), while \two proposing certified lane tips allows \sys{} to commit an arbitrarily large backlog with complexity independent of its size. We discuss next how \sys{} reaches consensus on a cut (\S\ref{sec:core_consensus}), and how it carefully exploits the lane structure to achieve seamlessness (\S\ref{sec:synchronization}). 


\subsubsection{Core consensus.}
\label{sec:core_consensus}

\sys{}'s consensus protocol follows a classic two-round (linear) PBFT-style~\cite{castro1999pbft} agreement pattern, augmented with a fast path that reduces latency to a single round (3mds) in gracious intervals.
The protocol progresses as a series of slots: each slot $s$ is assigned to a leader that can start proposing a new lane cut once slot $s-1$ commits. Within each slot, \sys{} follows a classical view based structure~\cite{castro1999pbft} -- when a slot's consensus instance fails to make timely progress, replicas revolt, triggering a \textit{View Change} and electing a new leader. 

Each view consists of two phases: \textit{Prepare} and \textit{Confirm} (5mds total). The Prepare phase tries to achieve agreement within a view (non-equivocation) while the Confirm phase enforces durability across views (Fig. \ref{fig:consensus}). During gracious intervals the \textit{Confirm} phase can be omitted (Fast Path), reducing consensus latency down to a single phase (3mds). 

\begin{figure}[h!]
    \vskip -8pt
   \centering
   \includegraphics[width=0.45\textwidth]{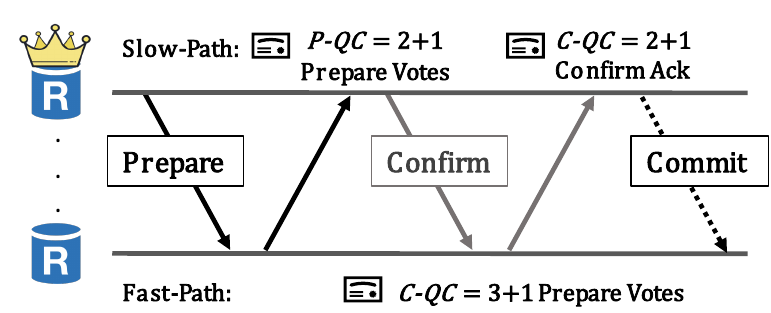}
   \caption{Core Consensus coordination pattern.}
   \label{fig:consensus}
   \vskip -0pt
\end{figure}

For simplicity, we first present our protocol for a single slot. We assume that all replicas start in view $v=0$. We explain \sys{}'s view change logic in \S\ref{sec:view_change}, and discuss in \S\ref{sec:parallel} how to orchestrate agreement for slots in parallel to reduce end-to-end latency.  We defer formal proofs of correctness and seamlessness to the Appendix, \ref{sec:proofs}.



\vskip -0pt
\begin{algorithm}
\caption{Consensus layer cheat sheet.}\label{alg:consensus-state}
\begin{algorithmic}[1]
\State{$CommitQC_s$}~~~~~~$\triangleright$ Commit proof for slot $s$
\State{$TC_{s, v}$}~~~~~~~~~~~~~~~~~$\triangleright$ View change proof for slot $s$, view $v$
\State{Ticket  $T_{s,v}$}~~~~~~~~~$\triangleright$ $\textsc{CommitQC}_{s-1}$ or $\textsc{TC}_{s, v-1}$
\State{Proposal  $P_{s,v}$}~~~~~$\triangleright$ Proposed lane cut. Requires $T_{s,v}$
\State{\textit{prop}}~~~~~~~~~~~~~~~~~~$\triangleright$ map: slot $\rightarrow$ $P_{s,v}$ 
\State{\textit{conf}}~~~~~~~~~~~~~~~~~~$\triangleright$ map: slot  $\rightarrow$ \textsc{PrepareQC}$_{s, v}$ 
\State{\textit{last-commit}}~~~~~~~$\triangleright$ map: lane $\rightarrow$ pos \label{var:lane}
\State{$log$}~~~~~~~~~~~~~~~~~~~~~$\triangleright$ total order of all user requests
\end{algorithmic}
\end{algorithm}
\vskip -0pt

\par \textbf{Prepare Phase (P1).}

\par \protocol{\textbf{1: L} $\rightarrow$ \textbf{R}: Leader $L$ broadcasts \textsc{Prepare}.}
A leader $L$ for slot $s$ begins processing (view $v=0$) once it has received a ticket $T \coloneqq \textsc{CommitQC}_{s-1}$ (Alg.~\ref{alg:consensus-state}), and ``sufficiently'' many new data proposals to propose a cut that advances the frontier committed in $s-1$; we coin this requirement \textit{lane coverage}, and defer discussion of it to \S\ref{sec:coverage}.

The leader broadcasts a message $\textsc{Prepare} \coloneqq (\langle P \rangle_L, T) $ that contains the \textit{consensus proposal} $P$ as well as the necessary ticket $T$. The proposal $P$ itself consists of the slot number $s$, the view $v$, and a cut of the latest certified lane tips observed by $L$:
$P \coloneqq \langle s, v, [lanes[1].tip.\textsc{PoA}, ..., lanes[n].tip.\textsc{PoA}]\rangle$.
\par \protocol{\textbf{2: R} $\rightarrow$ \textbf{L}: Replica $R$ processes \textsc{Prepare} and votes.}
$R$ first confirms the \textsc{Prepare}'s validity: it checks that \one the ticket is for the right leader/slot), and \two $R$ has not yet voted for this slot. $R$ then stores a copy of the proposal ($\textit{prop}[s] = P$), and responds with a $\textsc{Prep-Vote} \coloneqq \langle dig=h(P) \rangle_R$. Finally, if it has not locally received all tips (and their histories) included in $P$, $R$ begins to asynchronously fetch all missing data (\S\ref{sec:synchronization}). 
This is possible thanks to the non-blocking sync guarantee of the data dissemination layer: $R$ knows that the data will be available, and thus need not sync before voting.

\par \protocol{\textbf{3: L}: Leader $L$ assembles quorum of PREP-VOTE messages.}
$L$ waits for at least $n-f = 2f+1$ matching \textsc{Prep-Vote} messages and aggregates them into a \textit{Prepare Quorum Certificate (QC)}
$\textsc{PrepareQC}$ $\coloneqq$ ($s$, $v$, $dig$, $\{\textsc{Prep-Vote}\}$). This ensures agreement \textit{within} a view. No two \textsc{PrepareQC}s with different $dig$ can exist as all \textsc{PrepareQC}s must intersect in at least one correct node, who will never vote more than once per view.

\par \textbf{Fast Path.} In the general case, the Prepare phase is insufficient to ensure that a particular proposal will be persisted across views. During a view change, a leader may not see sufficiently many \textsc{PREP-VOTE} messages to repropose this set of operations. In gracious intervals, however, the leader $L$ can receive $n$ such votes by waiting for a small timeout beyond the first $n-f$. A quorum of $n$ votes directly guarantees durability across views (any subsequent leader will observe at least one \textsc{PREP-VOTE} for this proposal). No second phase is thus necessary.
In this scenario, $L$ upgrades the QC into a \textit{fast commit quorum certificate} and proceeds directly to the Commit step, skipping the confirm phase: it broadcasts the \textsc{CommitQC}, and commits locally. 
\par \protocol{\textbf{4 (Fast): L} $\rightarrow$ \textbf{R}: $L$
upgrades to \textsc{CommitQC} and commits.}
\vspace{2pt} 
\par \textbf{Confirm Phase (P2).} In the slow path, the leader $L$ instead moves on to the Confirm phase.
\par \protocol{\textbf{1: L} $\rightarrow$ \textbf{R}: Leader $L$ broadcasts \textsc{Confirm}.}
$L$ forwards the $\textsc{PrepareQC}_{s,v}$ by broadcasting a message $\textsc{Confirm} \coloneqq \langle PrepareQC_{s,v}  \rangle$ to all replicas.
\par \protocol{\textbf{2: R} $\rightarrow$ \textbf{L}: Replica $R$ receives and acknowledges \textsc{Confirm}.}
$R$ returns an acknowledgment $\textsc{Confirm-Ack} \coloneqq \langle dig \rangle_R$, and buffers the \textsc{PrepareQC} locally ($\textit{conf}[s] = \textsc{PrepareQC}$) as it may have to include this message in a later view change.

\par \protocol{\textbf{3: L} $\rightarrow$ \textbf{R}: Leader $L$ assembles \textsc{CommitQC} and commits.}
$L$ aggregates $2f+1$ matching \textsc{Confirm-Ack} messages into a \textit{Commit Quorum Certificate} $\textsc{CommitQC}$ $\coloneqq$ ($s$, $v$, $dig$, \{\textsc{Confirm-Ack}\}), and broadcasts it to all replicas.

%

\subsubsection{Processing committed cuts.} 
\label{sec:synchronization}

Recall that consensus proposals in \sys{} are cuts of data proposals, rather than simple batches of transactions. Upon committing, the replicas must then establish a (consistent) total order of proposals to correctly execute transactions. To help with this process, each replica maintains a log of all ordered transactions, and a map \textit{last-commit} which tracks, for each lane, the position of the latest committed data proposal (Alg. \ref{alg:consensus-state}).

\par \textbf{Seamless data synchronization.}\label{sec:seamless-sync}
To process a cut, a replica must first ensure that, for each lane, it is in possession of all data proposals that are transitively referenced by the respective tip (\textit{tip history} for short). Otherwise, it must first acquire (\textit{synchronize}) the missing data; in the worst-case, this may be the entire history. During a partial partition, for instance, lanes may grow without all correct replicas being aware, requiring them to later synchronize on missing data. 

\sys{} exploits the lane structure to synchronize in parallel with consensus, and in just a single message exchange, no matter the size of the history.
First, certified tips allow replicas to transitively prove the availability of a lane history \textit{without} directly observing the history. This allows replicas with locally inconsistent lane states to vote for consensus proposals without blocking (\textit{non-blocking sync}), and moves synchronization off the timeout-critical path. If \sys{} employed only best-effort (non-certified) lanes, replicas would have to synchronize \textit{before} voting on a consensus proposal to assert whether a tip (and its history) indeed exists. Such blocking increases the risk of timeouts undermining seamlessness.
Second, in-order voting in the data layer enables \textit{timely sync} by guaranteeing that there exists at least one correct replica that is in possession of \textit{all} data proposals in a certified tip's history.

To synchronize, a replica requests (from the replicas that voted to certify the tip) to SYNC all proposals in the tip's history with positions greater than its locally last committed lane proposal (up to $\textit{last-commit}[tip.lane].pos +1$). 
By design, lane positions must be gap free, and proposals chained together; the requesting replica thus knows \one whether a SYNC-REPLY contains the correct number of requested proposals, and \two whether the received proposals form a correct (suffix of the) history. By correctness of consensus, all correct replicas agree on the last committed lane proposal, and thus synchronize on the same suffix.

\par \textbf{Creating a Total Order.}
Once a replica has fully synchronized the cut ($\langle tips \rangle$) committed in slot $s$, and all previous slots have been committed and processed, it tries to establish a total order across all ``new'' data proposals subsumed by the lane tips. 
A replica first identifies, for each lane $l$, the oldest ancestor $start_l$ of $tips[l]$ with a position exceeding $last_l \coloneqq \textit{last-commit}[l]$ ({\em i.e.,} $start_l.pos == last_l+1$). 
Next, a replica updates the latest committed position for each lane ($\textit{last-commit}[l] = tips[l].pos$), and appends to the log all new proposals (lane proposals from $start_l$ to $tips[l]$) using some deterministic zipping function. 

We note that the proposal committed at $last_l$  may not (and need not) be the parent of $start_l$; for instance, the replica governing $l$ may have equivocated and sent multiple proposals for the same position. This does not affect safety as we simply ignore and garbage collect ancestors of $start_l$.

\subsubsection{Governing Progress.}\label{sec:coverage}
Correct leaders in \sys{} only initiate a new consensus proposal once data lanes have made ``ample'' progress. 
We coin this requirement \textit{lane coverage} -- it is a tunable hyperparameter that governs the pace of consensus. Intuitively, there is no need for consensus if there is nothing new to agree on!

Choosing the threshold above which there are sufficiently many new data proposals to agree on requires balancing latency, efficiency, and fairness. \changebars{\sys{}'s parameterizable lane coverage, unlike contemporary protocols that impose a fixed strategy~\cite{castro1999pbft, yin2019hotstuff, danezis2022narwhal, Bullshark}, offers deployments a \textit{choice}:}{}
a system with heterogeneous lane capacities, for instance, may opt to minimize latency by starting consensus for every new car; a resource conscious system, in contrast, may opt to wait for multiple new data proposals to be certified.\fs{in this case, one may consider replacing intermittent cars (all but the tip) by best effort broadcasts to further reduce costs.} 

In \sys{}, we set coverage, as default, to require at least $n-f$ new tips. This ensures that at least half of the included tips belong to correct replicas. \changebars{Tips, notably, may subsume a different number of proposals. Rather than enforce that replicas must make even progress (at times, elsewhere referred to as {\em chain quality}~\cite{danezis2022narwhal}), \sys{} opts to provide each replica the same \textit{opportunity} for progress and thus can organically adapt to heterogeneous lane growth. }{} \fs{Note: It's not always a good thing to force everybody to propose. Some replicas are slower or have less to propose. Flexible adaptation to heterogeneity is important}

\fs{ADDED BACK}
Lane coverage in \sys{} is best-effort and not enforced. From a performance perspective, a Byzantine leader that disregards coverage is \changebars{equivalent to}{indistinguishable from} a crashed one; from a fairness perspective, it is equivalent to leaders in traditional BFT protocols that assemble batches at will. The next correct leader (in a good interval) will make up for any lost coverage.

\par \textbf{Reliable Inclusion.}
Importantly, lane coverage does not cause \sys{} to discriminate against slow lanes. Unlike contemporary DAG-BFT protocols~\cite{danezis2022narwhal, Bullshark} that are driven by the fastest $n-f$ replicas, and may ignore proposals of the $f$ slowest, proposal cuts in \sys{} include \textit{all} $n$ lanes. 
Since correct replicas' lanes never fork, all of their proposals are guaranteed to commit in a timely fashion (as soon as a correct leader receives a \textsc{PoA} that subsumes them). 
\fs{could cut this?}
\changebars{}{This property is a direct consequence of \sys{}'s careful separation of concern between the data layer and the consensus layer. The data layer is \textit{exclusively} responsible for disseminating data and making it available to the consensus layer. The consensus layer is instead responsible for committing cuts of \textit{all the data}. As such, the $n-f$ fastest replicas still drive consensus, but necessarily include all data that has been received, without penalizing slow nodes.} 

Further, \sys{}'s consensus layer ensures that \textit{all} certified proposals--including those of Byzantine replicas--are reliably committed\changebars{}{, or in case of forks, garbage collected}\nc{Saves a line and we don't really talk about gc elsewhere}\fs{we do mention it right at the end of the preceeding section (Creating Total Order)}. Since lane growth requires that at least one correct replica be in possession of a \textsc{PoA} for history of the latest car, Byzantine proposers cannot continue to disseminate transactions without intending to commit them; thus effectively bounding the maximum amount of ``waste'' that correct replicas may receive from Byzantine actors. We defer a formal analysis to the Appendix (\ref{sec:bounded-waste}).

\subsection{View change}
\label{sec:view_change}
In the presence of a faulty leader, progress may stall. As is standard, \sys{} relies on timeouts in these situations to elect a new leader: replicas that fail to see progress revolt, and create a special \textit{Timeout Certificate} (\textsc{TC}) ticket to trigger a \textit{View Change}~\cite{castro1999pbft, kotla10zyzzyva}. 
Each view $v$ (for a slot $s$) is mapped to one designated leader (e.g. using round-robin). To prove its tenure a leader carries a ticket $T_{s,v}$, corresponding to either (in $v=0$) a $\textsc{CommitQC}_{s-1}$,\footnote{The view associated with a \textsc{CommitQC} ticket is immaterial; by safety of consensus, all \textsc{CommitQC}'s in a slot $s$ must have the same value.} or (in $v>0$) a quorum of mutineers $TC_{s, v-1}$ from view $v-1$. 

\textsc{Prepare} and \textsc{Confirm} messages contain the view $v$ associated with the sending leader's tenure. Each replica maintains a $\textit{current-view}_s$, and ignores all messages (besides \textsc{CommitQC}'s) with smaller views; if it receives a valid message in a higher view $v'$, it buffers it locally, and reprocesses it once it reaches view $v'$.
Per view, a replica sends at most one \textsc{Prep-Vote} and \textsc{Confirm-Ack}. 


A replica starts a timer for view $v$ upon first observing a ticket $T_{s,v}$. Seamlessness alleviates the need to set timeouts aggressively to respond quickly to failures. \sys{} favors conservative timers for smooth progress in good intervals as the data layer continues progressing even during blips.
The replica cancels the timer for $v$ upon locally committing slot $s$ or observing a ticket $T_{s, v+1}$. 
Upon timing out, it broadcasts a complaint message \textsc{Timeout} 
and includes any locally observed proposals that could have committed. 

\par \protocol{\textbf{1: R} $\rightarrow$ \textbf{R}: Replica $R$ broadcasts \textsc{Timeout}.}
A replica $R$, whose timer $t_v$ expires, broadcasts a message $\textsc{Timeout} \coloneqq \langle s, v, highQC, highProp \rangle_R$.
$highQC$ and $highProp$ are, respectively, the $\textsc{PrepareQC}_{s,v'}$ ($conf[s]$) and proposal$_{s, v''}$ ($prop[s]$) with the highest views locally observed by $R$. $R$ will ignore future \textsc{Prepare} and \textsc{Confirm} messages in view $v$.



\par \protocol{\textbf{2: R} Replica $R$ forms TC and advances view.}
A replica $R$ accepts a \textsc{Timeout} message if it has not yet advanced to a higher view ($\textit{current-view}_R \leq \textsc{Timeout}.v$) or already received a \textsc{CommitQC} for that slot. In the latter case, it simply forwards the \textsc{CommitQC} to the sender of \textsc{Timeout}. 

Replica $R$ joins the mutiny for view $v$ once it has accepted at least $f+1$ \textsc{Timeout} messages (indicating that at least one correct node has failed to see progress) and broadcasts its own \textsc{Timeout} message for $v$. This ensures that if one correct replica locally forms a \textsc{TC}, all correct replicas eventually will. 
Upon accepting $2f+1$ distinct timeouts for $v$ a replica assembles a \textit{Timeout Certificate} \textsc{TC} $\coloneqq (s, v, \{\textsc{Timeout}\})$, and advances its local view $\textit{current-view} = v+1$ and starts a timer for $v+1$. If $R$ is the leader for $v+1$, it additionally begins a new Prepare phase, using the \textsc{TC} as ticket. 




A leader $L$ that uses a \textsc{TC} as ticket must recover the latest proposal that \textit{could have} been committed at some correct replica.
To do so, the leader replica chooses, as \textit{winning proposal}, the greater of \one the highest $highQC$ contained in \textsc{TC}, and \two the highest proposal present $f+1$ times in \textsc{TC}; in a tie, precedence is given to the $highQC$. This two-pronged structure is mandated by the existence of both a fast path and a slow path in Autobahn. If a proposal appears $f+1$ times in \textsc{TC}, then this proposal may have gone fast-path as going fast-path requires $n$ votes.  Any later $n-f$ quorums is guaranteed to see at least $f+1$ such votes. Similarly, if a \textsc{CommitQC} formed on the slow path (from 2f+1 \textsc{Confirm-Acks} of a \textsc{PrepareQC}), at least one such \textsc{PrepareQC} will be included in the \textsc{TC} (by quorum intersection). 
Finally, the leader reproposes the chosen consensus proposal.

Put together, this logic guarantees that if in some view $v$ a \textsc{CommitQC} was formed for a proposal $P$, all consecutive views will only re-propose $P$.


\par \protocol{\textbf{3: L} $\rightarrow$ \textbf{R}: Leader $L$ sends $\textsc{Prepare}_{v+1}$.}
Finally, the leader $L$ broadcasts a message $\textsc{Prepare} \coloneqq (\langle v+1, \textit{tips} \rangle_L), \textsc{TC})$, where \textit{tips} is the winning proposal derived from TC; or $L's$ local certified lane cut if no winning proposal exists. Replicas that receive such a Prepare message validate \textsc{TC}, and the correctness of $tips$, and otherwise proceed with Prepare as usual. 
\subsection{Parallel Multi-Slot Agreement}
\label{sec:parallel}

Until now, we have described \sys{} as proceeding through a series of sequential slot instances. Unfortunately, lane proposals that narrowly "miss the bus" and are not included in the current slot's proposed consensus cut may, in the worst case, experience the sequential latency of up to two consensus instances (the latency of the consensus instance they missed in addition to the one necessary to commit).

To circumvent sequential wait-times, \sys{}, inspired by PBFT~\cite{castro1999pbft}, allows the next slot instance to begin in parallel with the current one, without waiting for it to commit.
We modify a slot ticket for $s$ to no longer require a $\textsc{CommitQC}_{s-1}$: the leader of slot $s$ can begin consensus for slot $s$ as soon as it receives the first \textsc{Prepare}$_{s-1}$ message (the view is immaterial), and has observed sufficiently many new tips \changebars{}{(w.r.t. to the tip cut in \textsc{Prepare}$_{s-1}$)} to satisfy lane coverage. 

Parallel consensus can avoid sequential delays entirely when instantiated with a stable leader (the same leader proposes multiple slots).  When instantiated with rotating leaders, there remains a single message delay: the time to receive $\textsc{Prepare}_{s-1}$.
We adjust \sys{} as follows.

\par \textbf{Managing concurrent instances.}
\sys{} maintains independent state for all ongoing consensus slot instances. Each instance has its own consensus messages, and is agnostic to all other processing. 
Once a slot has committed, it waits until all preceding slot instances have already executed to execute itself. Similarly, when ordering a slot's proposal (a cut of tips), \sys{} replicas identify and order only proposals that are new. Old tip positions are simply ignored. This  filtering step is necessary as  consecutive slots may propose non-monotonic cuts.
\par \textbf{Adjusting view synchronization.}
To account for parallel, overlapping instances, replicas no longer begin a timer for slot $s$ upon reception of a $\textsc{CommitQC}_{s-1}$, but upon reception of the first $\textsc{Prepare}_{(s-1, v)}$ message.\fs{timeouts should account for time it might take to forward Prepare to the leader that needs it + time to get coverage}
Moreover, we modify the leader election scheme to ensure that different slots follow a different leader election schedule. Specifically, we offset each election schedule by $f$. Without this modification, the worst case number of sequential view changes to commit $k$ successive slots would be $\frac{k * (f+1)}{2}$, as each slot $s$ would have to rotate through the same faulty leaders to generate a ticket for $s+1$.\fs{note: this is meant to be as: If slot s starts with R1, then slot s+1 starts with slot R1+f; for f=1 this is normal rotating leader}

\par \textbf{Bounding parallel instances}
During gracious intervals, the number of concurrent consensus instances will be small.  As each consensus instance incurs less than five message delays, there should, in practice, be no more than three to four concurrent instances at a time.  During blips, however, this number can grow arbitrarily large: new instances may keep starting (as they only need to observe a relevant \textsc{Prepare} message). Continuously starting (but not finishing) parallel consensus instances is wasteful; it would be preferable to wait to enter a good interval, as a single new cut can subsume all past proposals.
\sys{} thus bounds the maximum number of ongoing consensus instances to some $k$ by (re-) introducing \textsc{CommitQC} tickets.
To begin slot instance $s$, one must include a ticket $\textsc{CommitQC}_{s-k}$.  When committing slot $s$ we can garbage collect the \textsc{CommitQC} for slot $s-k$ as $CommitQC_s$ transitively certifies commitment for slot $s-k$. 

\fs{Bonus argument: If we only use Prepare as ticket, then the CommitQC for each slot must be kept around in order to reply to view changes. e.g. to help inform replicas that haven't committed). We'd like to bound this. Can do so either by active checkpointing (requires synchronization), or by leveraging the depth limiter} 

\par \textbf{Discussion.} \sys{}'s parallel proposals departs from modern chained-BFT protocols~\cite{giridharan2023beegees, yin2019hotstuff, jolteon-ditto} which pipeline consensus phases (e.g. piggyback $\textsc{Prepare}_{s+1}$ on $\textsc{Confirm}_s$). We decide against this design for three reasons: \one pipelining phases is not latency optimal as it requires at least two message delays between proposals, \two it introduces liveness concerns that require additional logic to circumvent~\cite{giridharan2023beegees}, and \three it artificially couples coverage between successive slots, {\em i.e.,} coverage for slot 2 affects progress of slot 1. 
\subsection{Optimizations}
Finally, we discuss a number of optimizations, that, while not central to \sys{}'s ethos, can improve performance.

\fs{Where to best put this?}
\subsubsection{Pipelining cars.} Replicas in \sys{} disseminate a lane proposal at position $pos$ only upon completing certificate for $pos-1$. This is not strictly necessary for seamlessness, and \sys{} can, akin to consensus (\S\ref{sec:parallel}), be configured to pipeline cars. Doing so, however, introduces a tradeoff: pipelining cars increases the number of un-certified proposals in the system that might never commit (\ref{sec:bounded-waste}).\fs{or require synchronization with optimistic tips!} We choose not to include this optimization in our prototype. \fs{alternatively to pipelining cars within a lane, one could also create "sub-lanes" for each lane, and stagger those (e.g. start proposal only upon receiving prev proposal... }

\subsubsection{Exploring \textit{un-certified} tips.}
\label{sec:optimistic_tips}
By default, \sys{} consensus pessimistically only proposes \textit{certified} tips. This allows voting in the consensus layer to proceed without blocking, which is crucial for seamlessness (\S\ref{sec:consensus_on_lanes}). Assembling and sharing certificates, however, imposes a latency overhead of three message exchanges in addition to the unavoidable consensus latency. This results in a best-case end-to-end latency of six message delays on the fast path, and eight otherwise. \sys{} employs two optimizations to reduce the latency between the dissemination of a data proposal and it being proposed as part of consensus (inclusion latency).  
\par \textbf{Leader tips.} First, we allow a leader to include a reference to the latest proposal ($dig, pos$) it has broadcast as its own lane's tip, even as the tip has not yet been certified.
If a Byzantine leader intentionally does not disseminate data, yet includes  the proposal as a tip, it is primarily hurting itself. Correct nodes will, at most, make a single extra data sync request before requesting a view change.  Moreover, if the Byzantine leader ever wants to propose additional data, it will have to backfill the data for this proposal (in-order voting).

\par \textbf{Optimistic tips.} Second, \sys{} may choose to relax seamlessness and allow leaders to optimistically propose non-certified tip references ($dig, pos$) for lanes that have, historically, remained reliable.\footnote{Akin to recent works~\cite{carousel, spiegelman2023shoal}, \sys{} can be augmented with simple reputation mechanisms to determine reliable lanes (\ref{s:reputation}).} In gracious intervals, this optimization lowers inclusion latency to a single message exchange, which is optimal. However, this comes at the risk of incurring blocking synchronization and triggering additional view changes. Replicas that have not received a tip locally hold off from voting, and request the tip from the leader directly. 

\fs{Not voting because of lack of availability unnecessarily correlates availability and agreement: Instead of not voting, cast a "weak-vote" only for agreement.}

We note, that critical-path synchronization is only necessary for the tip itself \fs{and uncertified parents if using pipelined cars} (and thus has constant cost) -- by design, all ancestors must be certified, and thus can be synchronized asynchronously. 
Further, synchronization is only necessary in view $v=0$, or in subsequent views that did not select a winning proposal from a prior view. In all other cases, the selected proposal $P$ included in \textsc{Prepare}$_{v' > 0}$ is already implicitly certified (at least $f+1$ replicas voted for $P$).

In some cases, holding off from voting may be overly pessimistic and can unnecessarily impede agreement: while an individual replica might not (yet) be in possession of a tip's data, there might be sufficiently many (at least $f+1$) other replicas that are, and can guarantee availability; yet insufficiently many (at least $2f+1$) to reach a \textsc{PrepareQC}. To account for this, \sys{} can refine the voting procedure by allowing replicas to explicitly cast a "weak" \textsc{Vote} for agreement only, or a "strong" \textsc{Vote} that asserts both agreement and local availability. A replica may cast both a weak and strong \textsc{Vote} for the same value within a view. A \textsc{PrepareQC} is valid if it consists of $2f+1$ \textsc{Vote}'s, at least $f+1$ of which must be strong; thus ensuring availability. Fast \textsc{CommitQC}s, however, must consist of $3f+1$ strong \textsc{Vote}'s; this ensures that view change tickets retain at least $f+1$ strong \textsc{Vote}s.

\fs{longer version in comment:}

\fs{old, longer version in comments}

\vspace{2pt}
\subsubsection{Other optional modifications.}
\fs{could possibly be cut/replaced with a reference to appendix. If cut, then consolidate the section intro.}
\sys{} can, if desired, be further augmented with several standard techniques. We discuss them briefly for completeness.

\vspace{1pt}
\par \textbf{Signature aggregation.}
\sys{}, like others~\cite{gueta2019sbft, yin2019hotstuff, Bullshark} can be instantiated with threshold (TH) signatures~\cite{boldyreva2002threshold, shoup2000practical} that aggregate matching vote signature of any Quorum Certificate (PoA, PrepareQC, CommitQC) into a single signature. This reduces \sys{}'s complexity per car and consensus instance (during good intervals) to $O(n)$; view changes require $O(n^2)$ messages~\cite{yin2019hotstuff}. 

\vspace{1pt}
\par \textbf{All-to-all communication.}
\sys{} follows a linear-PBFT communication pattern~\cite{gueta2019sbft} that allows it to achieve linear complexity (during good intervals) when instantiated with aggregate signatures. \sys{} can, of course, also be instantiated with all-to-all communication for better latency (but no linearity). In this regime, consensus consists of only 3 message exchanges (2mds on FastPath): replicas broadcast \textsc{Prepare-Vote}s and \textsc{Confirm-Ack}s, and respectively assemble $PrepareQC$s and $ConfirmQC$s locally. 
\fs{Slightly more discussion in comments:}

\fs{I've cut ride sharing here, and moved it entirely to Appendix.}
\begin{figure*}[!th]
\centering
\begin{minipage}{.6\textwidth}
  \centering
  \includegraphics[width=1\linewidth]{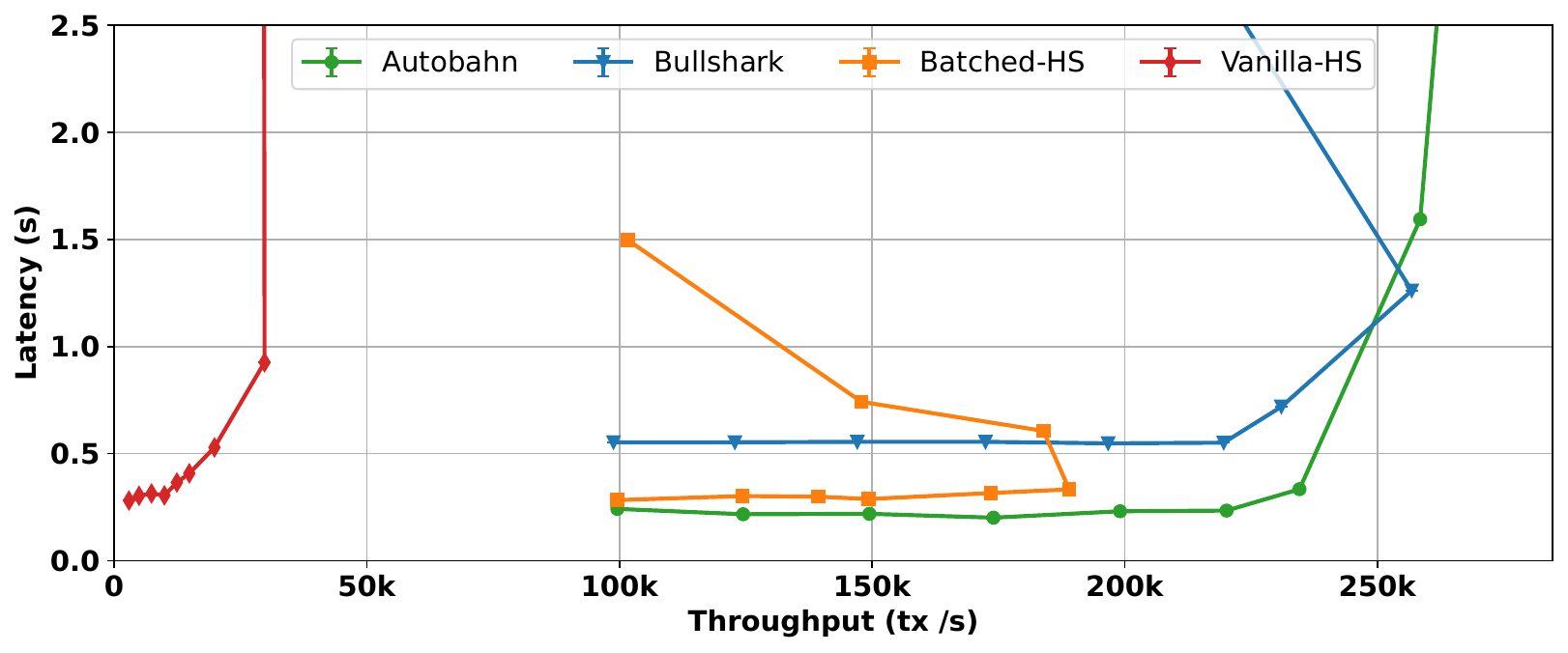}
  \vskip -2pt
  \caption{Throughput and Latency under increasing load.}
  \label{fig:performance}
\end{minipage}
\begin{minipage}{.39\textwidth}
  \centering
  \includegraphics[width=1\linewidth]
  {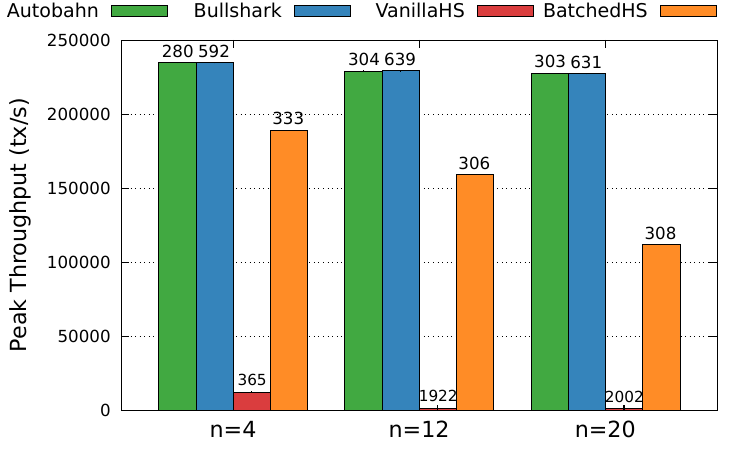}
  \vskip -0pt
  \caption{Peak throughput for varying $n$. Numbers atop bars show measured latency (ms).}
  \label{fig:Scaling}
\end{minipage}
\vskip -10pt
\end{figure*}

\section{Evaluation}
\label{sec:eval}




Our evaluation seeks to answer three questions:
\vspace{-1pt}
\begin{enumerate}
    \item \textbf{Performance}: How well does \sys{} perform in gracious intervals? (\S\ref{sec:perf})
    \item \textbf{Scaling} How well does it scale as we increase $n$? (\S\ref{sec:perf})
    \item \textbf{Blip Tolerance}\fs{or: Robustness}: Does \sys{} make good on its promise of seamlessness? (\S\ref{sec:async})
\end{enumerate}
\vspace{-1pt}


We implement a prototype of \sys{}\footnote{https://github.com/neilgiri/autobahn-artifact} in Rust, starting from the open source implementation of Narwhal and Bullshark~\cite{codebullshark}. We use Tokio's TCP~\cite{tokio} for networking, and RocksDB~\cite{rocksdb} for persistent storage. Finally, we use ed25519-dalek~\cite{ed25519} signatures for authentication. 

\par \textbf{Baseline systems.} We compare against open-source prototypes of HotStuff~\cite{yin2019hotstuff} and Bullshark~\cite{Bullshark, spiegelman2022bullshark} which, respectively, represent today's most popular traditional and DAG-based BFT protocol.
We use the "batched" HotStuff (BatchedHS) prototype~\cite{codehotstuff} that simulates an upper bound on achievable performance by naively separating dissemination and consensus. Replicas optimistically stream batches, and consensus leaders propose hashes of any received batches.\fs{This improves both throughput and latency as \one one proposal may reference batches from multiple replicas, and \two batches can be proposed before the initiator becomes the leader itself.} This design is not robust in real settings: \one Byzantine replicas and blips may exhaust memory. There is no guarantee that batches will be committed (no \textit{reliable inclusion}) requiring the system to eventually drop transactions. \two Replicas must synchronize on missing batches before voting, which can cause blips (no \textit{non-blocking, timely sync}). We also evaluate a standard version of HotStuff (VanillaHS) in which we modify BatchedHS to send batches only alongside consensus proposals of the issuing replica.

\fs{is this paragraph necessary overall? I tried to slightly shorten}
Both Bullshark~\cite{codebullshark} and \sys{} adopt the horizontally scalable worker layer proposed by Narwhal~\cite{danezis2022narwhal}. For a fair comparison with HotStuff we run only with a single, co-located worker; in this setup the Reliable Broadcast used by workers is unnecessary, so we remove it to save latency.
\nc{Seems fine to me}




\par \textbf{Experimental setup.} We evaluate all systems on Google Cloud Platform (GCP), using an intra-US configuration with nodes evenly distributed in \texttt{us-west1}, \texttt{us-west4}, \texttt{us-east1} and \texttt{us-east5}. We summarize RTTs between regions in Table~\ref{tab:latencyrtt}.
We use machine type \texttt{t2d-standard-16}\cite{gcp_machines} with 20GB of SSD, and 10GB/s network bandwidth.
A client machine (co-located in the same region as the replica) issues a constant stream of \changebars{no-op}{} transactions (tx), consisting of 512 random bytes~\cite{danezis2022narwhal, Bullshark}; we denote as (input) load the total number of tx/s streamed by all clients. 

We measure latency as the time between transaction arrival at a replica and the time it is ready to execute; throughput is measured as execution-ready transactions per second. \changebars{Our measurements do not include client-to-replica (and replica-to-client) latency; this latency is constant, small (compared to the overall consensus latency), and independent of the protocol choice. Client request invocation requires only an in-data center ping latency to the local replica ($\approx 15\mu s$). Replies to the client incur the tail latency of the $f+1$st closest replica ({\em i.e.,} 19/28ms in our setup), or only the ping latency if the local replica is trusted (as is common in blockchains).}{} 

We set a batch size of 500KB (1000 transactions) across all experiments and systems, but allow consensus proposals to include/reference more than one batch if available; this \textit{mini-batching} design~\cite{codehotstuff,codebullshark} allows replicas to organically reach larger effective batch sizes with reduced latency trade-off.\fs{Fix? We cap HS to some large max? Exact number Neil?} 
All systems are run with rotating leaders; we set view timers to 1s.\fs{We run BatchedHS capped at X because after that it stops working (too much sync); in practice you want to cap it anyways to avoid also byz synchronization abuse} 
Each experiment runs for 60 seconds.

\begin{table}[t]
    \vskip 4pt 
    \centering
    \footnotesize{
    \begin{tabular}{|c|c|c|c|c|}
           RTT & us-east1 & us-east5 & us-west1 & us-west4 \\
           us-east1 & 0.5  & 19 & 64 & 55 \\
           us-east5 & 19 & 0.5 & 50 & 57\\
          us-west1 & 64 & 50 & 0.5 & 28\\
          us-west4 &55 & 57 & 28 & 0.5 \\
    \end{tabular}
    }
    \vskip 2pt
    \caption{RTTs between regions (ms)}
    \label{tab:latencyrtt}
    \vskip -3pt 
\end{table}

\fs{It seems like about 5MB payload size is where a broadcast bottlenecks. That means that 10 digests per replica are needed }

\begin{figure*}[!t]
\hspace*{-0.5cm}   
\centering
\begin{minipage}{0.48\textwidth}
  \centering
  \includegraphics[width=1\linewidth]{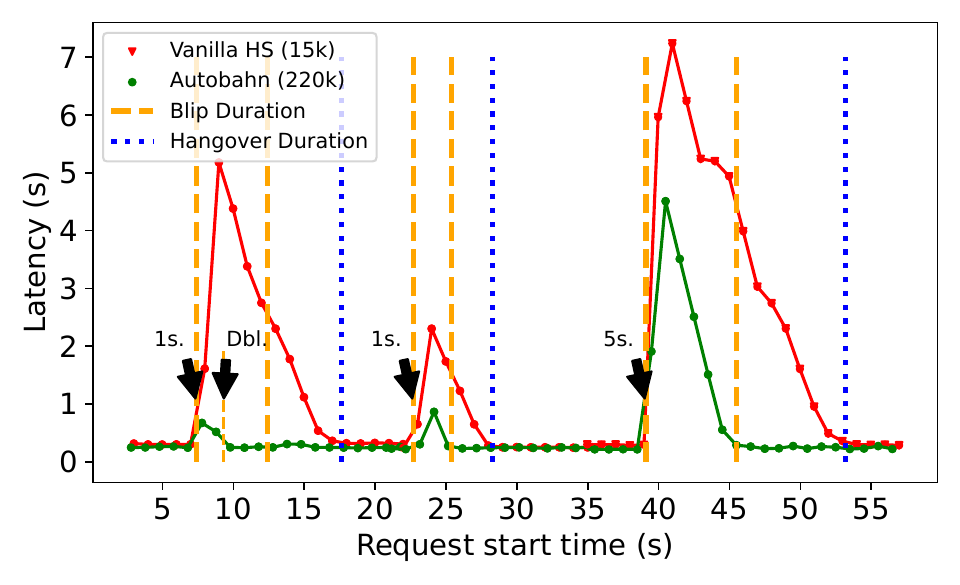}
    \vskip -3pt
  \caption{Leader failures.\fs{make figures more rectangular}}
  \label{fig:leader-fail}
\end{minipage}
\hspace{4mm}
\begin{minipage}{0.48\textwidth}
  \centering
    \includegraphics[width=1\linewidth]{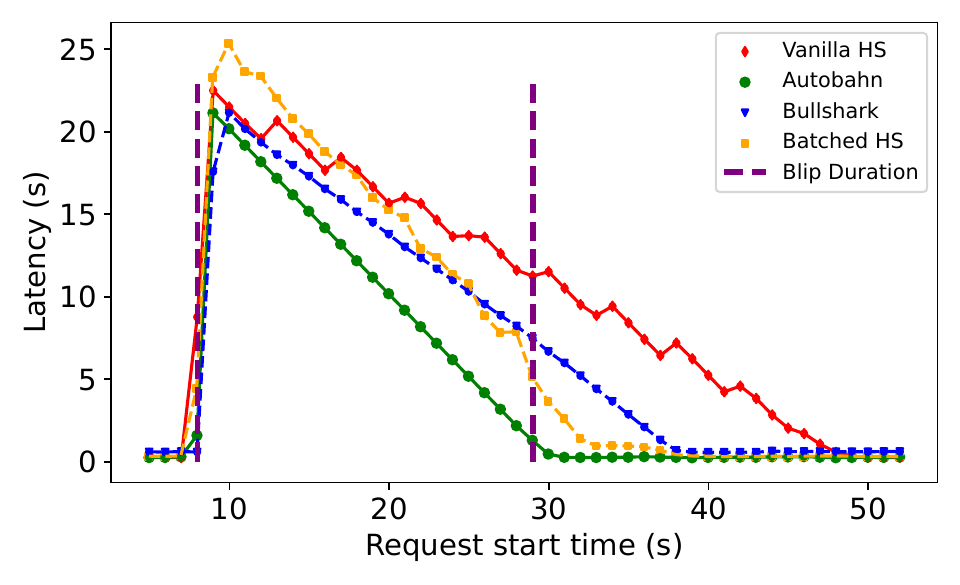}
  \vskip -3pt
  \caption{Partial partition.\fs{OUr line wrong: still workig on fix...}}
  \label{fig:partition}
\end{minipage}
\vskip -10pt
\end{figure*}

\subsection{Performance in ideal conditions}
\label{sec:perf}
\fs{note: we ran all experiments with debug prints enabled, but that should hurt us the most. Running release mode only made about a 2ms difference for us, so we keep the numbers.}
Figure \ref{fig:performance} shows the performance in a fault-free, synchronous setting, using $n=4$ nodes.\fs{hopefully we can still get around to some bigger ones!} 
\sys{}'s throughput matches that of Bullshark (ca. 234K tx/s), but reduces latency by a factor of 2.1x (from 592ms to 280ms), edging out even BatchedHS (333ms) and VanillaHS (365ms).
HotStuff and Bullshark require, respectively 7mds and up to 12 mds (10.5~mds average) to commit a transaction, while \sys{} reduces coordination to 4mds on the Fast Path, and 6mds otherwise (with optimistic tips).
Both \sys{} and Bullshark significantly outperform both HotStuff variants for throughput by a factor of 1.23x (BatchedHS) and 15.6x\fs{I tentatively adjusted HS tput upwards to 15k, because that's what Neil used for the blip experiments (just to be consistent)} (VanillaHS), and are bottlenecked on the cost of deserializing and storing data on disk.
VanillaHS offers very low throughput \changebars{(ca. 15K tx/s before latency noticeably rises)}{(12.4k tx/s)} as it requires that replicas disseminate their own batches alongside consensus proposals. The system quickly bottlenecks on the (network and processing) bandwidth of one broadcast. \fs{it's probably not "just" the bandwidth, but a combination of receiver side processing too: e.g. deserialization, persistent storage, etc.}
BatchedHS scales significantly better (189K tx/s) than VanillaHS as it amortizes broadcasting cost. At high load, however, we find that it incurs a significant amount of data synchronization (resulting at times even in view changes) causing it to bottleneck; this corroborates findings in Narwhal~\cite{danezis2022narwhal} and demonstrates that this design is not robust in practice. \fs{alt: demonstrating that the design is not robust in practice (corroborating~\cite{danezis2022narwhal} findings).}

\par \textbf{Fast Path/Optimistic Tips.} To better understand the performance benefits of the fast path and optimistic tip optimizations, we evaluate \sys{} with the fast path deactivated and without the optimistic tip optimization. We omit the graph for space constraints. We observe a 40ms increase in latency when forced to always go ``slow''. This increase is smaller than a cross-country RT: the fast path incurs higher tail latency due to the larger CommitQC ($|n|$) and because the leader ends up contacting non-local replicas, including those located across  coasts. Similarly, we record an additional 33ms increase when proposing only certified tips.  This follows directly from the need to contact $f+1$ replicas to form a PoA. 

\par \textbf{Scaling number of replicas.} Figure \ref{fig:Scaling} shows the peak throughput achievable as we increase the number of replicas. 
Both \sys{} and Bullshark scale gracefully as $n$ increases; throughput remains bottlenecked on data processing, while latency is unaffected. In our setup, the total peak load $P$ remains constant, but each replica proposes a smaller fraction as $n$ grows since input load is balanced across replicas ($P/n$ per replica).\footnote{If replicas were to instead propose at constant rates $r < P/n$, throughput would increase with $n$ (up to the data processing bottleneck $P$).}
BatchedHS, in contrast, experiences a sizeable throughput degradation with growing $n$ ($-16\%$ at $n=12$, and $-41\%$ at $n=20$) as replicas are increasingly likely to be out of sync and thus are forced to perform additional data synchronization\changebars{}{on the critical path}\fs{save orphan?}.\fs{it's quite surprising actually. Since there are no faults, any synchronization necessary is simply from momentary asynchrony, and overheads of data processing; us and Bullshark seem to avoid this, through what I can only explain as a "buffer-time" introduced by cars/RB respectively.} Finally, we observe that VanillaHS suffers a significant throughput-latency tradeoff: because each replica proposes unique batches, and batches are only proposed when a replica becomes the leader, latency grows proportionally to $n$. We highlight this tension by bounding latency to be at most 2s. In this setup, VanillaHS throughput diminishes sharply as $n$ rises; it drops from \changebars{15K}{12.4k} to only 1.5K tx/s as latency surges rapidly. VanillaHS is 1.3 times slower than Autobahn at n=4, but over 6.6 times slower at n=20.

\subsection{Operating in \textit{partial} synchrony}

\label{sec:async}
The previous section highlighted \sys's good performance in the failure-free scenario: \sys{} achieves the best of both worlds. It matches Bullshark's throughput while preserving HotStuff's low latency. We now investigate whether \sys{} remains robust to blips. We use $n=4$. 

\par \textbf{Faulty Leader Blips.} We first consider blips that result from a single leader failure. Bullshark and BatchedHS exhibit the same blip behavior as \sys{}, and we thus omit them for clarity. Figure~\ref{fig:leader-fail} plots latency over time. Data points are averages over second-long windows. For the first two blips, we use a standard timeout of 1s~\cite{priv-com-aptos-mysten}.\fs{Timeouts in VanillaHS are per phase, whereas they are per (slot) view for Autobahn. To meaningfully compare blip behaviors, we use the same timeouts for both.}
 HotStuff's pipelined, rotating leader design, can, in the presence of a faulty leader, trigger not one but two timeouts. Because the votes for a proposal issued by a leader $L1$ are eagerly forwarded only to $L2$, a single failure of $L2$ can cause two timeouts to trigger: one for its preceding proposal, and one for its own. When instantiated with stable leaders, it can only generate a single blip. We thus consider both scenarios (respectively called Dbl, and 1s). VanillaHS (under 15K tx/s load) suffers from hangovers that persist for respectively 1.6x, 1.3x, and 1.2x of the blip duration beyond the good interval resuming. The system remains bottlenecked by data dissemination and delivery when trying to work off the transaction backlog. For the Dbl (3s) blip, for instance, the latency penalty not only exceeds the blip duration, for a maximum transaction latency of 6.2s, but remains as high as 1.3s after three seconds post-blip, and 700ms at four seconds post-blip (respectively 3.5x and 1.9x increase over the steady state latency of 365 ms).  \sys{} (under 220K tx/s load), in contrast, immediately recovers from the hangover. 
This behavior continues to hold true for larger timeout values. For a timeout as high 5s (still realistic -- Diem, Facebook's former blockchain used a timeout of 30s~\cite{baudet2019state, priv-com-aptos-mysten}), VanillaHS's hangover persisted for almost 7 seconds after the blip ended.

\nc{Old version in iffalse}

\par \textbf{Partition Blips.} Next, we consider a blip caused by a temporary network partition (20s) that isolates replicas into two halves (Fig. \ref{fig:partition}). The load is 15K tx/s (the max for VanillaHS).
\fs{Not: we ran without optimistic tips (shouldn't really affect latency)}\sys{} experiences only a small hangover because it continues to disseminate data (replicas can reach $f+1$ other replicas, enough to grow their lanes). Once the partition resolves, a single slot instantly commits the entire lane backlog, and replicas identify and request (in a single step) all missing data. It takes, respectively, about 1s to transfer and process all missing data (we are bottlenecked by bandwidth and delivery); this illustrates an unavoidable, non protocol-induced hangover. Transactions submitted during the \changebars{partition}{asynchronous} period experience a latency penalty close to the remaining blip duration. Bullshark and BatchedHS require slightly longer to recover (both ca. 9s). BatchedHS optimistically continues to disseminate during the partition, but upon resolution must enforce a cap on mini-batch references per proposal to avoid excessive synchronization on the timeout-critical path. Bullshark cannot continue to disseminate during the partition because it must reach $2f+1$ nodes to advance the DAG; after it resolves, however, its efficient dissemination can quickly recover the backlog.
VanillaHS, in contrast, exhibits a large hangover, directly proportional to the blip duration. 
\fs{BatchedHS (at 15k tx/s) experiences an 8s hangover: it continues to optimistically disseminate during the partition, but must synchronize on missing data on the timeout-critical path. To avoid continuous view changes, we must cap the maximum number of mini-batch references per proposal to 15. This ensures progress, but introduces a slight hangover. }

\fs{Our hangover can probably be smoothened out if we allowed to commit txs as soon as everything preceeding them in the total order has arrived. However, in our prototype we first wait for ALL txs in the history we are currently committing, and THEN we order and execute them all. }
\fs{NOTE: our latency is slightly higher than HS here. It's probably because the load is not high enough, causing us to wait on batch timers. The problem is that if we reduce our batch timers things break again...}
\fs{we can instantly commit a large lane, (and very quickly get all the car headers), but the data itself can take a while to get...}
\nc{leaving as is until final results come in}

\fs{note: not entirely sure why it zigzags, but it's probably because commitment can happen "out of order". For example, if a Tx1 arrives at R1 at time T, and Tx2 arrives at R2 at time T'>T during the partition, but when the partition ends R1 is the last to become leader, Tx2 may commit first}
\fs{TODO: add bullshark and BatchedHS. Because the partition separates west/east, Bullshark is not able to get any $2f+1$ quorums and thus stops advancing entirely (consensus and data dissemination are tightly coupled). Upon the end of the partition it thus has a hangover, but can recover faster since it disseminates more efficiently. The same applies to BatchedHS -> Maybe it actually ends up having a big hangover because consensus keeps timing out.}

\fs{old in comments}

\fs{more optional experiments in comments:}
\section{Related Work}
\label{sec:related}
\sys{} draws inspiration from several existing works. Scalog~\cite{ding2020scalog} demonstrates how to separate data dissemination from ordering to scale throughput at low latencies (for crash failures). Narwhal~\cite{danezis2022narwhal} discusses how to implement a scalable and BFT robust asynchronous broadcast layer, a core tenet for achieving seamlessness. \sys{}'s consensus itself is closely rooted in traditional, latency-optimal BFT-protocols, most notably PBFT's core consensus logic and parallelism~\cite{castro1999pbft}, Zyzzyva's fast path~\cite{kotla10zyzzyva, gueta2019sbft}, and Aardvark's doctrine of robustness~\cite{aardvark}. 

\par \textbf{Motorizing Consensus Protocols.} \sys{}'s architecture, by design, isn't tied to a specific consensus mechanism, and thus can easily accommodate future algorithmic improvements (including advances in asynchronous protocols).
Furthermore, existing consensus deployments may easily be \textit{motorized} to enjoy both improved throughput and seamlessness without adopting the custom consensus protocol described in \sys{}: a blackbox consensus mechanism of choice may be augmented by adding lanes, replacing consensus proposals with cuts, and adopting the synchronization and ordering logic presented.

\par \textbf{Traditional BFT} protocols such as PBFT~\cite{castro1999pbft}, HotStuff~\cite{yin2019hotstuff}, and their respective variants~\cite{kotla10zyzzyva, gueta2019sbft}\cite{jalalzai2020fasthotstuff, jolteon-ditto}
optimize for good intervals: they minimize latency in the common\changebars{}{(synchronous)} case, but guarantee neither progress during periods of asynchrony, nor resilience to hangovers. Multi-log frameworks~\cite{stathakopoulou2019mir, stathakopoulou2022state, gupta2021rcc, arun2022scalable} partition the request space across multiple proposers to sidestep a perceived leader bottleneck, and intertwine the sharded logs to arrive at a single, total order. 
They inherit both the good and the bad of their BFT building blocks: consensus has low latency, but remains prone to hangovers. 

\par \textbf{Directed Acyclic Graph (DAG) BFT} protocols, as previously discussed~\cite{baird2016swirlds, keidar2021all, danezis2022narwhal, malkhi2022maximal, Bullshark}, originate in the asynchronous model and enjoy both excellent throughput and network resilience; recent designs \cite{spiegelman2022bullshark, spiegelman2023shoal, keidar2023cordial} make use of partial synchrony to improve consensus latency.
\textit{Certified} DAGs~\cite{danezis2022narwhal, Bullshark} are close to achieving seamlessness (data synchronization on the timeout-critical path aside), but use Reliable Broadcast~\cite{bracha1985asynchronous, bracha1987asynchronous} (RB) to implement the voting steps of traditional BFT consensus, resulting in high latencies. Bullshark~\cite{spiegelman2022bullshark}, for instance, requires up to 4 rounds of RB (each comprised of 3 message exchanges) to commit, for a worst case total latency of 12mds (9mds when optimized by Shoal~\cite{spiegelman2023shoal}).

Recent works propose (but do not empirically evaluate) un-certified DAGs~\cite{keidar2023cordial, malkhi2024bbca}. These designs aim to improve latency by replacing the RB in each round with best-effort-broadcast (BEB), but forgo seamlessness because they must synchronize (possibly recursively) on missing data before voting.

Star~\cite{duan2024dashing} adopts a construction similar to \sys{}. It uses $n$ parallel proposers that create "weak-certificates" that only prove availability (equivalent to PoA), but retains a rigid round-based DAG-like construction. For each round it requires an external single-shot PBFT instance that commits $n-f$ weak-certificates; this design is prone to hangovers as the number of required consensus instances is directly proportional to blip duration.

\par \textbf{DAG vs Lanes}. Lanes in \sys{} loosely resemble the DAG structure, but minimize quorum sizes and eschew rigid dependencies across proposers. 
\changebars{Notably, \sys{} proposes a cut of all $n$ replica lanes; in contrast, DAGs can reliably advance proposals simultaneously from only $n-f$ replicas, and are thus subject to ignoring some proposals~\cite{Bullshark}.}{DAGs must operate at the proposal rate of the $n-f$'th fastest replica ({\em i.e.,} slowest correct, in presence of faults), and are subject to orphaning slow replicas' proposals~cite{Bullshark}. In contrast, \sys{}'s independent lanes allow for flexible proposal rates (each replica can propose at their load and resource capability), and consensus guarantees to commit all $n$ replica lanes (no orphans).} \fs{note: previously we said "slowest correct replica", but that was a bit imprecise and a reviewer asked}

\fs{DAGs bigger quorums also increase tail latencies and make it less robust to partitions?}

\par \textbf{Asynchronous BFT consensus} protocols~\cite{vaba, danezis2022narwhal, jolteon-ditto, liu2023flexible, honeybadger} sidestep the FLP impossibility result~\cite{fischer1985impossibility} and guarantee liveness even during asynchrony by introducing randomization. 
They typically use a combination of reliable broadcast (RB)~\cite{bracha1987asynchronous, bracha1985asynchronous} and asynchronous Byzantine Agreement (BA) to implement asynchronous Byzantine atomic broadcast~\cite{honeybadger, beat, gao2022dumbo}.
Unfortunately, these protocols require expensive cryptography and several rounds of message exchanges, resulting in impractical latency.\fs{(Threshold-signatures (for coin generation) and -encryption (to handle network adversary))} \neil{complexity measures for single-shot async (mvba) in message delays: vaba - 13-19.5 md, dumbo mvba 19-47 md, speeding dumbo 6-12 md. For multi-shot need to add in 3 md for RB in beginning}
The asynchronous systems closest related to \sys{} are Honeybadger~\cite{honeybadger} and Dumbo-NG~\cite{gao2022dumbo}. 
Honeybadger~\cite{honeybadger} achieves high throughput by operating $n$ instances of RB in parallel and trading off latency for large batch sizes; it uses $n$ binary BA instances~\cite{bracha1987asynchronous, mostefaoui2014signature, abraham2022efficient, cachin2000random} to select $n-f$ proposals that may commit. Dumbo-NG~\cite{gao2022dumbo} separates data dissemination from consensus to ameliorate batching/latency trade-offs. It constructs parallel chains of RB and layers a \changebars{}{sequential multi-shot} multi-valued BA protocol~\cite{vaba, cachin2001secure} atop; this architecture closely resembles the spirit of \sys{}, but incurs significantly higher latency.\fs{which proposes cuts of RB certificates}


\section{Conclusion}
\label{sec:conclusion}
This work presents \sys{}, a novel partially synchronous BFT consensus protocol that is actually robust to \textit{partial synchrony}. 
\sys{} achieves seamlessness, while simultaneously matching the throughput of DAG-based BFT protocols and the latency of traditional BFT protocols.
Drive safe!


\section*{Acknowledgments}
We are grateful to Brad Karp (our shepherd), and the
anonymous reviewers for their thorough and insightful
comments. This work was supported in part by
the NSF grants CSR-17620155, CNS-CORE 2008667, CNS-CORE 2106954 \fs{these were from Basil -change}, a Sui Research Award, support from ithe Initiative for Cryptocurrencies and Contracts (IC3) as well as gifts from Accenture, AMD, Anyscale, Google,
IBM, Intel, Mohamed Bin Zayed University of Artificial Intelligence,
Samsung SDS, SAP, and VMware.

\bibliographystyle{ACM-Reference-Format}
\bibliography{References/references}

\appendix
\pagebreak \newpage


\section{Proofs}
\label{sec:proofs}


\subsection{Safety}\label{proof:safety}
We show that all correct replicas commit, for every slot, the same proposal. It necessarily follows, that all correct replicas compute a consistent log. In the following, we prove agreement for a given slot $s$, and drop the slot number for convenience.

We first prove that, within a view $v$, no two conflicting proposals can be committed. 
\begin{lemma}\label{lem:non-equiv}
    If there exists a \textsc{PrepareQC} or \textsc{CommitQC} for consensus proposal $p$ in view $v$, then there cannot exist a \textsc{PrepareQC} or a \textsc{CommitQC} for consensus proposal $p'\neq p$ in view $v$. 
\end{lemma}

\begin{proof}
    Suppose for the sake of contradiction there exists a \textsc{PrepareQC} or \textsc{CommitQC} for consensus proposal $p'\neq p$ in view $v$. 
    There must exist at least $n-f$ replicas who sent a \textsc{Prep-Vote} for $p'$ in view $v$. At least $n-f$ replicas sent a \textsc{Prep-Vote} for proposal $p$ in view $v$. These two quorums intersect in at least $1$ correct replica, a contradiction since a correct replica only sends one \textsc{Prep-Vote} message in a view.
\end{proof}

\begin{lemma}\label{lem:same-view}
    If a correct replica commits a (consensus) proposal $p$ in view $v$, the no correct replica commits a proposal $p'\neq p$ in view $v$.
\end{lemma}

\begin{proof}
    Suppose for the sake of contradiction a correct replica commits proposal $p'\neq p$ in view $v$. There must exist a fast \textsc{CommitQC} or a \textsc{PrepareQC} for $p'$ in view $v$. By lemma \ref{lem:non-equiv}, however, any \textsc{PrepareQC} or fast \textsc{CommitQC} in view $v$ must be for proposal $p$.
\end{proof}
Next, we show that if a proposal committed in some view $v$, no other proposal can be committed in a future view.
\begin{lemma}\label{lem:prepare}
    If a correct replica commits a consensus proposal $p$ in view $v$, then any valid \textsc{Prepare} for view $v'>v$ must contain proposal $p$.
\end{lemma}

\begin{proof}
    We prove by induction on view $v'$.

    \textbf{Base case:} Let $v'=v+1$. Suppose for the sake of contradiction there exists a valid \textsc{Prepare} for a proposal $p'\neq p$ in view $v'$. There are two cases: either 1) $p$ was committed in view $v$ on the fast path or 2) $p$ was committed in view $v$ on the slow path.
    
    \textbf{1) Fast path.} Since $p$ was committed on the fast path there must exist $n$ \textsc{Prep-Vote} messages for $p$ in view $v$, of which at least $n-f$ correct replicas stored a \textsc{Prepare} for proposal $p$ locally. Any TC in view $v$ contains $n-f$ \textsc{Timeout} messages, and by quorum intersection at least $n-2f$ of these \textsc{Timeout} messages must contain a \textsc{Prepare} message for proposal $p$. Since a \textsc{Timeout} message can contain at most $1$ \textsc{Prepare} message, there cannot exist $n-2f$ \textsc{Prepare} messages for $p'$ in view $v$. By lemma \ref{lem:non-equiv}, any \textsc{PrepareQC} or \textsc{CommitQC} in view $v$ must be for $p$. Therefore any \textsc{PrepareQC}, \textsc{CommitQC}, or set of $n-2f$ \textsc{Prepare} messages for $p'$ must be in a view $<v$. Therefore, the winning proposal for any TC in view $v$ must be $p$, a contradiction since by the assumption there was a valid \textsc{Prepare} containing a valid TC for proposal $p'$.

    \textbf{2) Slow path.} Since $p$ was committed on the slow path there must exist $n-f$ \textsc{Confirm-Ack} messages (\textsc{CommitQC}) for $p$ in view $v$, of which at least $n-2f$ correct replicas stored a \textsc{PrepareQC} for proposal $p$ locally. Any TC in view $v$ contains $n-f$ \textsc{Timeout} messages, and by quorum intersection at least $1$ of these \textsc{Timeout} messages must contain a \textsc{PrepareQC} for proposal $p$. By lemma \ref{lem:non-equiv}, any prepare or \textsc{CommitQC} in view $v$ must be for $p$. Therefore any \textsc{PrepareQC} or \textsc{CommitQC} for $p'$ must be in a view $<v$. Therefore, the winning proposal for any TC in view $v$ must be $p$ since even if there does exist $f+1$ \textsc{Prepare} for $p'$ in view $v$, a prepare/\textsc{CommitQC} takes precedence. This, however, is a contradiction since by the assumption there was a valid \textsc{Prepare} in view $v'$ containing a valid TC with a winning proposal $p'$.
    
    \textbf{Induction Step:} We assume the lemma holds for all views $v'-1> v$, and now prove it holds for view $v'$. Suppose for the sake of contradiction there exists a valid \textsc{Prepare} for a proposal $p'\neq p$ in view $v'$. 
    Since $v'>0$, this \textsc{Prepare} message must contain a TC for view $v'-1$. This TC consists of $n-f$ \textsc{Timeout} messages, for which the winning proposal is $p'$. This means that the highest (by view) \textsc{PrepareQC} or set of $n-2f$ matching \textsc{Prepare} messages must be for $p'$. By the base case and induction step any valid \textsc{Prepare} (and thereby \textsc{PrepareQC}) in views $>v$ must be for proposal $p$.  There are two cases: either $p$ was committed on the fast path or on the slow path in view $v$.
    
    \textbf{Fast Path.} Since $p$ was committed on the fast path, the only \textsc{Prepare} message all correct replicas voted for in view $v$ was for proposal $p$, so there can exist at most $f$ \textsc{Timeout} messages containing a \textsc{Prepare} in view $v$ for proposal $p'$. Thus, any \textsc{PrepareQC} or set of $n-2f$ \textsc{Prepare} for proposal $p'$ must be in a view $<v$. Since $p$, was committed on the fast path, there exists at least $n-f$ replicas who updated their highest \textsc{Prepare} to be \textsc{Prepare} for proposal $p$ in view $v$. By quorum intersection, any TC must contain at least $n-2f$ of these messages, a contradiction since the winning proposal is for $p'$.
    
    \textbf{Slow Path.} By lemma \ref{lem:non-equiv}, any \textsc{PrepareQC} in view $v$ must be for proposal $p$. Thus, any \textsc{PrepareQC} must be in a view $<v$. Since $p$ was committed on the slow path, there must exist at least $n-f$ replicas who voted for a \textsc{PrepareQC} for proposal $p$ in view $v$.By quorum intersection, any TC must contain at least $1$ \textsc{Timeout} message from this set of replicas, which will contain a \textsc{PrepareQC} with view at least $v$. By the induction assumption any \textsc{Prepare} message with view $>v$ must be for proposal $p$. Any ties in the view between a \textsc{PrepareQC} and a set of $f+1$ \textsc{Prepare} is given to the \textsc{PrepareQC}, so the winning proposal must be $p$ for any TC, a contradiction.
\end{proof}

\begin{lemma}\label{lem:safe-slot}
    If a correct replica commits a consensus proposal in slot $s$, then no correct replica commits a different consensus proposal in slot $s$.
\end{lemma}

\begin{proof}
   Let view $v$ be the earliest view in which a correct replica commits a proposal $p$. Suppose for the sake of contradiction a correct replica commits a proposal $p'\neq p$ in view $v'$. There are two cases: 1) $v'=v$ and 2) $v'>v$.
   \textbf{Case 1.} By lemma \ref{lem:same-view} no correct replica will commit a proposal $p'\neq p$ in view $v$, a contradiction. 
   \textbf{Case 2.} Since $p'$ is committed by a correct replica, there must exist a \textsc{CommitQC} or a fast \textsc{CommitQC} for $p'$ in view $v'$, which implies there must exist a valid \textsc{Prepare} message for $p'$ in view $v'$. By lemma \ref{lem:prepare} any valid \textsc{Prepare} message in view $v'>v$ must be for proposal $p$, a contradiction.
\end{proof}

Finally, we show that all correct replicas arrive at a consistent log.
\begin{theorem}
    All correct replicas commit the same requests in the same order.
\end{theorem}

\begin{proof}
    Correct replicas order all consensus proposals by increasing slot numbers. By lemma \ref{lem:safe-slot}, for any slot $s$, correct replicas commit the same consensus proposal in slot $s$. By the collision resistance property of hash functions, all correct replicas will agree on the same history for each tip within a consensus proposal. And since all correct replicas use the same deterministic function that resolves forks in Byzantine lanes, zips the data, and constructs an ordering within a slot, all correct replicas will order the same blocks in the same order.
\end{proof}

\subsection{Liveness}\label{proof:liveness}
Next, we prove that all correct clients' requests will eventually be committed. We assume a continuous input stream of client requests to all correct replicas (and thus that all data lanes progress). We first prove that \sys{} consensus itself is guaranteed to make progress given synchrony, and then show that every correct replica's data proposal will eventually be committed. 

To do so, we lay some brief groundwork. We define GST to be the (unknown) time at which the network becomes synchronous, and $\Delta$ to be the (known) upper bound on message delivery time after GST. For simplicity we assume synchrony lasts forever after GST\footnote{In practice, synchrony is of course finite, and only needs to be "sufficiently" long.}. We set a replica's local consensus timeout value for any slot $s$ \fs{can be set to a conservative. - or is 10 the minimum?} to  $10\Delta$; this, we show, suffices to guarantee commitment in good intervals.


We begin by showing that every correct replica will eventually commit a slot $s$

To do so, we first prove that after GST, a correct leader will take a bounded amount of time to send a \textsc{Prepare} message. For any view $v>0$, this is trivial but for $v=0$ it may take time for a correct leader to satisfy the coverage rule. We show this takes at most $2\Delta$ after GST.
\begin{lemma}\label{lem:coverage}
    If slot $s$, view $v$, starts after GST and is led by a correct leader, it will take at most $2\Delta$ time after receiving a ticket to send a \textsc{Prepare} message in slot $s$, view $v$.
\end{lemma}

\begin{proof}
    There are two cases: 1) $v=0$ or 2) $v>0$.
    
    \textbf{Case 1.} For view $0$, a correct leader must have received a \textsc{Prepare} message in slot $s-1$ containing a proposal with $n$ tips, of which at least $n-f$ are from correct replicas. A proof of availability for correct tips will take at most $\Delta$ time to form after GST, and the next tip at the next greater height will take at most $\Delta$ time to reach all replicas. Therefore, by $2\Delta$, the leader will receive at least $n-f$ new tips, satisfying the coverage rule to propose a new cut.

    \textbf{Case 2.} For view $v>0$, a correct leader must have received a TC for view $v-1$ in slot $s$. If the winning proposal for a TC is $\neq \bot$, then the leader will immediately send a \textsc{Prepare} message for slot $s$, view $v$, containing the winning proposal (and the TC). Otherwise if the winning proposal is $\bot$, the leader will immediately send a \textsc{Prepare} message containing its own local set of $n$ tips as the consensus proposal.
    Therefore, it takes at most $2\Delta$ for a correct leader to send a \textsc{Prepare} message in slot $s$, view $v$.
\end{proof}

Next, we show that all correct replicas will reliably enter a common view, in bounded time. This ensures that all correct replicas will accept messages from the leader, and respond accordingly.
\begin{lemma}\label{lem:view-sync}
    For any slot $s$, let $v$ be a view after GST with a correct leader, and $t$ be the time the first correct replica enters $v$. All correct replicas will enter $v$ by time $t+4\Delta$.
\end{lemma}

\begin{proof}
    The earliest correct replica to enter view $v$ in slot $s$ must have received a \textsc{Prepare} message in slot $s-1$ or a TC in slot $s$, view $v-1$ (ticket). It will then forward this ticket to the leader of slot $s$, view $v$, which will arrive at the leader within $\Delta$ time. \fs{a TC, by design, should also just form at everyone}\neil{right this is just to bound the time replicas receive the same ticket} If $v=0$, by lemma \ref{lem:coverage} the leader's proposal will take at most $2\Delta$ time to form. Otherwise if $v>0$, then the leader And it then sends a \textsc{Prepare} message for slot $s$, which arrives at all correct replicas within $\Delta$ time, after which all correct replicas enter view $v$.
\end{proof}

Recall, that replicas only vote for a consensus proposal (\textsc{Prepare}) if they can locally assert to the availability of all data. This is trivial when proposing only cerified tips (no synchronization is necessary), but when proposing optimistic tips, replicas may need to first fetch missing data. We show that this takes only a bounded amount of time.
\begin{lemma}\label{lem:data-sync}
    If a correct replica receives a \textsc{Prepare} message in slot $s$, view $v$, from a correct leader after GST, it will take at most $2\Delta$ before it sends a \textsc{Prep-Vote} message.
\end{lemma}

\begin{proof}
    a correct replica will only vote for a \textsc{Prepare} message if it can assert the availability of every tip (and it's history) in the consensus proposal. If the tip is certified (\textsc{PoA}), no further work is necessary, and the replica can vote immediately. If it instead is not certified (optimistic) a replica must ensure that it has already received the associated data proposal. If it has, it may vote immediately.
    Otherwise, a replica will send to the leader a \textsc{Sync} message, requesting all missing optimistic tips, which will arrive at the leader by $\Delta$ time. A correct leader, by design, will only propose optimistic tips for which it has all the associated data proposal; thus it will reply to a sync request with the missing data proposal (batch + parent \textsc{PoA}), which takes at most $\Delta$ time to arrive. Notably, no sync is necessary for ancestors of the respective tips, as the parent \textsc{PoA} transitively proves it's availability. 
\end{proof}

Next, we show that all correct replicas will commit in a view led by a correct leader.
\begin{lemma}\label{lem:honest-commit}
    If there is some view $v$ in slot $s$ that is led by a correct leader after GST, then all correct replicas commit the correct leader's proposal in view $v$.
\end{lemma}

\begin{proof}
    By lemma \ref{lem:view-sync}, all correct replicas will receive a \textsc{Prepare} message in slot $s$, view $v'$, and enter view $v'$ within $4\Delta$ of the first correct replica entering $v'$. By lemma \ref{lem:data-sync}, all correct replicas will take at most $2\Delta$ time to send a \textsc{Prep-Vote} to the leader, who will receive these votes by $\Delta$ time and form a \textsc{PrepareQC} from these $n-f$ votes (or a fast \textsc{CommitQC} if it receives $n$ votes. The leader will then send a \textsc{Confirm} message if it formed a \textsc{PrepareQC} or a \textsc{Commit} message if it formed a fast \textsc{CommitQC} to all replicas, which will arrive within $\Delta$ time. If a correct replica receives a valid fast \textsc{CommitQC}, it will commit. Otherwise, correct replicas will send a \textsc{Conf-Vote} to the leader, which will arrive at the leader by $\Delta$ time. The leader will form a \textsc{CommitQC} from $n-f$ \textsc{Conf-Vote} messages, and then send a \textsc{Commit} message containing the \textsc{CommitQC} to all correct replicas, who will receive it by $\Delta$ time. Upon receiving a \textsc{Commit} message, a correct replica will commit. Since \fs{haven't said it is? Should say: "by setting it to 10, we can thus guarantee". 10 seems high?}\neil{added a sentence in the beginning of the liveness section about this} the local timeout is $10\Delta$, all correct replicas will commit in view $v'$ before timing out.
\end{proof}

Using the above lemma, we show that every correct replica will eventually commit in slot $s$.
\begin{lemma}\label{lem:main-live}
    Let view $v$ in slot $s$ start after GST. Every correct replica eventually commits a proposal in a view $v'\geq v$ in slot $s$.
\end{lemma}

\begin{proof}
   Since the number of Byzantine replicas is bounded by $f$, there exists some view $v'\geq v$ in slot $s$, such that $v'$ is led by a correct leader. By lemma \ref{lem:honest-commit}, all correct replicas will commit the correct leader's proposal in view $v'$.
\end{proof}

To construct the log, a replica must be able to synchronize on all data proposals subsumed by a consensus proposal cut.
\begin{lemma}\label{lem:data-avail}
   For any committed proposal $p$, a correct replica will be able to retrieve all data included (and transitively referenced) by the tips in $p$. 
\end{lemma}

\begin{proof}
    We distinguish two cases: \one The protocol is instantiated to use only certified tips for proposals, and \two the protocol is configured to use the optimistic tip optimization presented in \S \ref{sec:optimistic_tips}.
    
    \one: Since proposed tips are certified ($PoA$ exists and is provided) at least one correct replica must be in possession of the tip's data proposal \textit{and} all of its history (FIFO voting). A replica can thus directly synchronize on the entire tip history by contacting the quorum of $f+1$ replicas that certified the tip; at least one correct replica will reply.

    \two: A correct replica does not vote for a proposal unless it is in possession of the data for all proposed tips. Since a correct replica only commits upon receiving $n$ \textsc{Prep-Vote} messages (fast \textsc{CommitQC}) or $n-f$ \textsc{Confirm-Ack} messages (\textsc{CommitQC}), there must exist at least $f+1$ correct replicas who must have voted. Consequently, for every committed optimistic tip $\geq f+1$ correct replicas have stored the corresponding data. A replica missing a committed optimistic tip thus simply requests to sync from the Quorum of replicas that voted to commit.
    We remark that optimistic tips are only considered valid by correct replicas if their direct parent is certified. Since the parent is certified, it follows from \one that the entire tip history can be synchronized.

\end{proof}

Finally, we show that all data proposals from correct data lanes are eventually committed. It follows that all correct client requests are eventually committed and executed:


\begin{lemma}\label{lem:client}
    If a correct replica receives a client request it will include it in a car $c$, and all correct replicas will eventually add a tip subsuming $c$ to its local set of tips.
\end{lemma}

\begin{proof}
    When a correct replica receives a client request it will add it to its next data proposal, which becomes the new tip of its data lane. Before disseminating the data proposal, a correct replica must wait for a proof of availability ($PoA$) to form for its parent. 
    All messages sent by correct replicas will eventually arrive at all other correct replicas. 
    Since \sys{}'s data dissemination layer is free of timeouts, message exchanges simply run at the pace of the network. 
    A correct replica's data proposal will eventually reach all correct replica, and thus $\geq f+1$ correct replicas will vote for it, guaranteeing that a $PoA$ can be formed (and thus allowing the next car to begin). 
    When instantiated with optimistic tips, all correct replicas will adopt the data proposal as this lanes highest optimistic tip.
    When instantiated without, the proposing replica will form a $PoA$ and disseminate it. All correct replicas will adopt the certified data proposal as the highest tip. 
    Since correct replicas never equivocate their lanes never fork. Thus the tip transitively subsumes all previously disseminated cars.
\end{proof}

\begin{lemma}\label{lem:growth}
    Given a continuous stream of client requests, a correct replica's lane will grow infinitely.
\end{lemma}

\begin{proof}
    As established in lemma \ref{lem:client}, correct replicas will form new cars upon receiving client requests, and succeed in forming $PoA$s. Since we assume clients to submit a continuous (infinite) stream of input requests, correct replicas' lanes too  will grow continuously.
\end{proof}

\begin{lemma}\label{lem:lane}
    All data proposals of a correct replica's lane will be committed and executed at all correct replicas.
\end{lemma}

\begin{proof}
    There is an infinite number of available consensus slots. Under round robin assignment, there thus must be an infinite amount of correct replicas that are chosen as initial proposer for a slot. 
    It follows from lemma \ref{lem:client} that every correct proposer will eventually see, and thus propose, the highest lane tips from all correct replicas' lanes. By lemma \ref{lem:growth} correct lanes grow infinitely, and thus lane coverage will always become satisfied. Finally, since a correct replica never forks, any tip in a correct replica's lane transitively subsumes all cars the replica ever proposed. 

    By lemma \ref{lem:honest-commit} all correct leader proposals after GST will be committed. It follows that all cars in correct lanes will be committed. By lemma \ref{lem:data-avail} all committed data proposals can be retrieved by every correct replica. It follows that all correct replicas will commit and execute all data proposals proposed by correct replicas.
\end{proof}

\begin{theorem}
    Eventually any client request will be ordered and executed.
\end{theorem}

\begin{proof}
   We assume that each client request will be resent to a different replica if it has not received a response after some timeout. Eventually, some correct replica, $h$, will receive the client request and include it in a car. By lemma \ref{lem:lane}, all correct replicas' data proposals will be ordered and executed.
  
\end{proof}


\subsection{Seamlessness discussion} \label{proof:seamless}

\subsubsection{\sys{} is seamless}
Next, we illustrate that \sys{} is not only live, but also \textit{seamless}. Three design considerations are key to avoid hangovers: 

\one First, \sys{} allows data dissemination to proceed at the pace of the network despite consensus blips. While consensus may fail to make progress due to missed timeouts, data lanes keep growing.

\two Second, upon eventual return of a good interval, \sys{} guarantees that all certified disseminated proposals are swiftly committed: consensus proposals reference an arbitrary amount of disseminated proposals with "constant" cost (cost independent of the backlog size)\footnote{Data proposals in \sys{} are linear in the number of lanes. However, for each lane, only a single tip is needed to reference an arbitrarily long history.}, and missing data can be synchronized with constant protocol latency. 

\three Finally, asynchronously disseminated data can be fetched off the timeout-critical\footnote{Every action, or delay, that may contribute to the violation of a timeout is \textit{timeout-critical}.} path of consensus, and thus does not make the protocol more susceptible to violating timeouts.

Our proof of seamlessness does \textit{not} extend to the optimistic tip optimization presented in \S\ref{sec:optimistic_tips}. We provide some additional discussion at the end of this section.

\par \textbf{A note on Lane Coverage:} For simplicity, we ignore considerations for lane coverage in the following discussion. We assume that lane coverage is set reasonably (w.r.t to the given dissemination pace), and assume that after the end of a consensus blip lane coverage will always be satisfied for new consensus proposals led by correct replicas. This assumption is justified if data dissemination outpaces consensus, which is "trivially" the case in presence of a consensus blip. If the assumption does not hold, then there cannot exist a noteworthy backlog (w.r.t to lane coverage configuration), and thus we need not be concerned about hangovers.\\

We first show that data dissemination is responsive, {\em i.e.,} continues to make progress at the pace of the network. Put differently, \sys{} continues to disseminate client requests even in the presence of consensus blips.

\begin{proof}
This follows directly from Lemmas \ref{lem:client} and \ref{lem:growth}.
\end{proof}

We say that data proposal has been successfully disseminated if it has been certified, and the certificate has reached at least one correct replica.\footnote{With the optimistic tip optimization it is enough for just the data proposal (uncertified) to arrive at a correct replica.} 
   
We showed in Lemma \ref{lem:client} that all disseminated data proposals will be subsumed by tips once available to a correct replica.
Since \sys{} consensus proposals consist of tips, and tips transitively reference the entire lane history it follows:

\begin{corollary}\label{cor:constant}
    \sys{} can propose arbitrarily long lane histories with constant cost.
\end{corollary}

We show next that available tips will be proposed as soon as a correct replica (that has seen the tip) becomes the leader for a new \textit{fresh} slot. Fresh denotes a slot that is in view $0$, {\em i.e.,} for which the leader can propose a new tip cut of choice (we assume lane coverage is given). This is as early as any protocol can hope to propose disseminated data: if it has not arrived at any correct replica, then it is not visible (not successful dissemination), and we cannot expect it to be proposed.

\begin{lemma}
    All data proposals issued by correct replicas, that have successfully been disseminated to a correct replica $R$, will be proposed in the first available \textit{fresh} slot $s$ led by $R$. 
\end{lemma}

\begin{proof}
    A correct replica $R$ that is the leader for a fresh slot $s$, {\em i.e.,} $v=0$, will propose the latest available tips it has received from all replicas. Consequently, it will implicitly propose to commit \textit{all} data proposals transitively referenced by the lane history. Since correct replica lanes do not fork, this implies that all data proposals issued by correct replicas, that have been successfully disseminated, will be proposed.
\end{proof}

From our liveness proof it follows that the proposal is guaranteed to be committed after GST. 
\fs{Of course, it may fortuitously commit earlier too, which is strictly advantageous. Although you can argue that in the GST model, if something committed, GST must've "held"}


In Lemma \ref{lem:data-avail} we showed that eventually a replica will be able to retrieve all data required to commit. We remark that synchronization for certified tips is asynchronous, {\em i.e.,} off the timeout-critical path for consensus. Thus \sys{}'s synchronization does not make the consensus protocol more susceptible to blips.
Replicas can vote without having synchronized, but must wait to synchronize before locally ordering and executing committed requests.

Next, we show that \sys{} is not only guaranteed to be able to fetch all data, but will do so in constant number of message exchanges (2 total), allowing replicas to complete synchronization by the time they locally commit a slot $s$ (subject to data bandwidth constraints). 
This property allows \sys{} to be seamless, as the maximum number of message exchanges for synchronization is independent of the size of the backlog. This ensures that synchronization does not introduce protocol-induced delays to commitment. We note that, of course, synchronization speed is bound to the available network bandwidth; we discuss in \ref{sec:hangover-limits} the implications.

\begin{lemma}\label{lem:fast-sync}
   For any committed proposal, a correct replica will be able to synchronize on missing cars in constant time, off the timeout-critical path, and by the time it commits the proposal. \fs{after GST} 
\end{lemma}

\begin{proof}

To synchronize on a certified tip \changebars{}{or any ancestor of an optimistic tip,} a replica sends a \textsc{Sync} message to quorum of the $f+1$ replicas who voted to create the proof of availability (PoA), of which at least $1$ must be correct. This \textsc{Sync} message will arrive at this correct replica $\Delta$ time afterwards. Since correct replicas respect FIFO voting, this correct replica must have the data for all ancestors for the missing tip. It will then send a \textsc{Sync-Reply} message to the requesting replica containing all data for the missing tip and its ancestors. This message will arrive at the correct replica within another $\Delta$ time, taking in total $2\Delta$ for a correct replica to synchronize. 
 A correct replica starts synchronizing upon receiving a \textsc{Prepare} message, and commits at the earliest after $2$ more message delays (Fast Path), and latest after $4$ more message delays (Slow Path). Since synchronization requires 2 message delays all correct replicas will synchronize on the data \textit{roughly} by the time consensus finishes:
 
\end{proof}

We conclude that \sys{} is seamless. Consensus blips do not halt data dissemination, and upon return of a good interval, the entire backlog of disseminated data can be committed in one consensus proposal, with protocol complexity independent of the backlog, thus avoiding hangovers.

\par \textbf{A note on Optimistic Tips:} The optimistic tip optimization may introduce synchronization on the timeout-critical path of consensus, and is thus, according to our definition, not seamless. 
Nonetheless, we posit that it is, in practice, quite robust to blips as the amount of on-critical-path synchronization is limited to a single car per lane. Replicas need only ever synchronize on the data proposal associated with the optimistic tip before voting.
Since the ancestor of an optimistic tip is certified, synchronization of the remaining (possibly long) lane history is asynchronous.

If a correct replica is missing data for an optimistic tip, then by lemma \ref{lem:data-sync} it will take at most $2\Delta$ to retrieve the data. Thus, in practice, one may simply increase the timeout by a constant amount ($2\Delta$) to avoid blips caused by synchronization.

\subsubsection{Practical considerations, and non-protocol induced hangovers.}\label{sec:hangover-limits}
Synchronization requires only two message exchanges (a sync request, and a sync reply) to synchronize on arbitrarily long data lanes. 

\par \textbf{Asymmetric Latencies:} Due to assymetric latencies, the 2md required to synchronize may be slightly longer (or shorter!) than the 2md to commit depending on which replicas need to be contacted. Likewise, on the Slow Path, the fastest $n-f$ replicas may create and distribute a $CommitQC$ before the slowest $f$ have received \textsc{Prepare} and started to synchronize. However, at least $f+1$ correct replicas will have synchronized in time successfully, which is enough for output commit.
Such practical considerations are orthogonal to seamlessness: any resulting delays are \one not protocol-induced (not unique to \sys{} and not caused by its protocol mechanisms), and \two of little practical concern. Since delays are \one independent of the size of committed history the effects are exceedingly minor.

\par \textbf{Bandwidth Constraints:} The size of the sync reply is, of course, proportional to the length data lane and the amount of data that must be fetched. If network bandwidth is a limiting factor, then synchronization of a large history may implicitly cause a brief hangover. We submit that this not in conflict with seamlessness, as it is not consensus protocol-induced. All data must be disseminated at some point, and synchronization does not introduce more data dissemination than was originally intended, and beyond what a subset of correct replicas have already received. Synchronization simply catches up all correct replicas on data dissemination they should have completed but have not, e.g. due to volatile network behavior such as partitions; such synchronization is unavoidable. To minimize synchronization in less extreme network cases, on may employ forwarding schemes (e.g. gossip of disseminated data) to maximize the amount of data already available.
Further, large histories that need to be synced can be staggered, and sent in FIFO order at the bandwidth the network allows. \sys{} may commit proposals as they arrive, and does not need to wait for "all" of a history to arrive in order to start committing. For example, for a lane $l$, data proposals at position $s$ can already be ordered and executed before the data proposals for $s+1$ arrive.

Traditional BFT protocols avoid "synchronizing" many messages at once because it does not allow dissemination to advance beyond consensus. They end up disseminating the same amount of data eventually, but will grind to a halt completely in the presence of consensus blips. Notably, replicas not part of a vote quorum may have to synchronize on missing data too upon commit.

\subsubsection{Are DAGs seamless?}\label{sec:dag-seamless}
In the following, we briefly discuss the seamlessness of existing DAG protocols. We distinguish two classes of DAG protocols: 
\one Certified DAGs, such as Narwhal~\cite{danezis2022narwhal} or Bullshark~\cite{Bullshark} employ Reliable Broadcast (RB) in each round to disseminate data.
\two Uncertified DAGs, such as Cordial Miners~\cite{keidar2023cordial} or BBCA-Chain~\cite{malkhi2024bbca} employ only Best-Effort-Broadcast (BEB) for data dissemination. We note that these two designs are currently theoretical-only, and have not been empirically evaluated.\footnote{BBCA-Chain uses a PBFT-like primitive for leader DAG nodes. All other DAG nodes use BEB.}

\par \textbf{Synchronization on timeout-critical path.} All four of the above DAG protocols incur blocking synchronization on the timeout-critical path of consensus.
For uncertified DAGs this is fundamental. Since data dissemination is best effort, edges in DAG nodes simply consist of reference (hash digest) to a preceding (typically of the previous round\footnote{Most DAGs progress in a rigid round by round pattern, but some exceptions exist; BBCA-Chain, for instance, appears to not prescribe any particular structure to its non-leader nodes.}) data proposal. Analogous to optimistic tips in \sys{} (\S\ref{sec:optimistic_tips}), such uncertified edge references do not guarantee availability. Since a replica can only ascertain the availability of a edge by locally being in possession of the data, it must synchronize on missing edges before adopting (or voting for) a new DAG node to asure liveness. Such synchronization is thus on the timeout-critical path of consensus, and makes the protocol susceptible to violating timeouts. 
Surprisingly, we observed existing certified DAGs (specifically Narwhal and Bullshark) also perform synchronization on the timeout-critical path. Existing protocol descriptions (Bullshark~\cite{Bullshark}, Alg.2, L.52) and open-source implementations~\cite{codebullshark} require causal histories to be available locally \textit{before} voting. Consequently, synchronization may contribute to producing consensus blips. However, synchronizing on the timeout-critical path is not fundamental to certified DAGs. Round certificates ($2f+1$ votes) implicitly act as availability proofs (akin to $PoA$ in \sys{}). Thus replicas can be allowed to vote before synchronizing, and data fetching can be moved off the timeout-critical path (in the background).

\par \textbf{Recursive synchronization.} In DAGs, each node contains edges to \textit{some} subset of preceeding nodes. Typically, a node in round $r$ requires $2f+1$ references to nodes in round $r-1$. 
Because of this design, it is not possible to infer the full causal history from a single DAG node in constant time. For instance, given only a node in round $r$, it is impossible to infer which other replicas' nodes it referenced in some round $r-k$ ($k>1$). The causal history may be "arbitrarily" intertwined, and can only be inferred by recursively tracing back the path of the edges.
As a result, in order to synchronize on a node's history, a replica may need to issue several rounds of sequential data fetching, in each round learning more of the history and thus \one what needs to be fetched next, and \two from which replica to fetch it (there is no guarantee that one correct replica is in possession of the entire required history). Unfortunately, such unpredictable synchronization time may result in hangovers.

We speculate that Narwhal and Bullsharks choice to synchronize on the timeout-critical path is largely to avoid such behavior. In practice, synchronizing on the timeout-critical path may help Narwhal and Bullshark to streamline synchronization and make commit time more predictable.\footnote{In Cordial Miners correct replicas pessimistically forward locally received data to other replicas. This helps improve state homogeneity, and thus reduces need for sync, but introduces potentially redundant overheads.} This makes the DAG operate somewhat in lock-step (and grow slower), but still cannot fully sidestep the problem. Furthermore, doing so also results in a conundrum between synchronization on the timeout-critical path, and poor robustness/predictability of synchronization time.

Our communications with companies deploying DAGs~\cite{priv-com-aptos-mysten} support our claims. They report that in heterogenous deployments (which are the norm in decentralized systems), replicas frequently incur large amount of synchronization (on the timeout-critical path), resulting in timeouts and loss of consensus progress. To avoid such scenarios, recent work has focused on leveraging reputation schemes~\cite{spiegelman2023shoal} to try to elect as leaders only the fastest, and best connected replicas. While this can help to some extent in practice, it does not eliminate the fundamental shortcoming; truly seamless protocols, like \sys{}, provide a more principled solution. \sys{}'s lane based design, for instance, drastically simplifies reasoning about synchronization: it easy to infer what needs to be synced on (a lane prefix), and where to request it from (a single correct replica must have it all).

\fs{could give example, but don't think it's needed}
\subsection{Fork Garbage Collection and Bounded Wastage}\label{sec:bounded-waste}
Under round-robin leader assignment every correct leader eventually issues a proposal. Consequently, every certified data proposal ({\em i.e.,} $PoA$) that has been received by any correct replica will eventually be proposed, and thus committed (reliable inclusion). This guarantees liveness for correct replica proposals, but it also allows us to \one garbage collect forks, and \two bound the amount of data proposals a Byzantine replica can disseminate without intending to commit. 

In order to issue a new data proposal, a Byzantine replica must certify it's previous car and distribute the $PoA$. By the reliable inclusion property, a Byzantine replica thus cannot issue a new data proposal, without implicitly confirming to commit a data proposal at all prior positions. In short, it is guaranteed that a data proposal will be committed for all but the latest lane position $pos_{latest}$.

We note, that this does not guarantee that all certified data proposals with $pos < pos_{latest}$ will commit.
Since Byzantine replicas may equivocate, it is possible to produce $PoA$'s for several different proposals per position. 
Each correct replicas votes for (and stores) at most one data proposal per lane position, but different replicas may store different certified data proposals for the same position.

By the reliable inclusion argument, all positions $<pos_{latest}$ will eventually be committed. Commit will order at most one proposal per position, and replicas can discard all obsolete data proposals locally stored. Thus it is guaranteed that eventually all forks will be resolved and garbage collected.

Furthermore, correct replicas store at most one data proposal issued by a Byzantine replica that is not guaranteed to commit or garbage collected, {\em i.e.,} the data proposal at $pos_{latest}$.
Since each data proposal enforces a bounded batch size $b$, a correct replica may at most have to store $f * b$ transactions (one proposal from $f$ Byz Replicas) that will never commit or be garbage collected, thus bounding the amount of \textit{wastage} in the local mempool at any time. 

\textbf{Pipelined cars.} \sys{} can, if desired, be augmented to pipeline the certification process of consecutive proposals to minimize sequential wait times in the case that batches are filled exceedingly quickly. Like parallel consensus proposals, the amount of uncertified cars at a given time may be bounded by enforcing that a position $p$ cannot begin before $p-k$ has been certified.
Pipelining up to $k$ cars, in turn, allows a Byzantine replica to distribute up to $k$ uncertified proposals that are not guaranteed to be reliably included. The wastage bound follows accordingly: $k * f * b$ transactions might never be committed or garbage collected.


\section{Additional Technical Discussion}
In the following we discuss some optional modification that are not integral to the core \sys{} protocol. We do not implement or evaluate these, but find there is value in nontethless presenting them.

\subsection{Avoiding Synchronization with Optimistic Tips}\label{s:reputation}
The optimistic tip optimization presented in \S\ref{sec:optimistic_tips} can, in the presence of Byzantine actors, cause (a constant amount of) synchronization on the timeout-critical path. For instance, a Byzantine proposer may forward its batch only to a correct leader, and no other replicas; if the leader proposes the batch, it must aid other correct replicas with synchronization before they may vote, which may  potentially cause its view to time out.

To reduce unwanted synchronization in practice, \sys{} can be augmented with a simple (replica-local) reputation mechanism. A leader forced to synchronize on some lane $l$'s tip downgrades its local perception of $l$'s reputability. Below a threshold, the leader stops proposing optimistic tips for $l$, and simply falls back to proposing a certified tip. Each replica is thus responsible for warranting optimism for its own lane and consequently its own data proposals' latency. Reputation can be regained over time, e.g. for each committed car.

We remark that such a reputation scheme serves exclusively to minimize \textit{redundant} dissemination, and is not necessary for safety or liveness. Without, the worst case synchronization effort on the critical path remains constant; only the optimistic tips must be synchronized, {\em i.e.,} at most one data proposal per lane.

\subsection{Ride Sharing}
Consensus phases and cars share the same basic message pattern~\cite{provable_abraham} and can, in theory, be sent and signed together (\textit{ride-sharing}) to minimize the total number of messages (and associated authentication costs).  


Hower, it may not always be convenient to hitch a ride: Cars must slow down to accommodate the larger Quorum sizes of consensus, and consensus messages may need to wait for available cars (up to 2md if one narrowly misses a \textsc{Proposal}. Ride sharing thus poses a trade-off between optimal resource utilization and latency. 

We find that in practice that, even for small $n$, consensus coordination is not a throughput bottleneck (data processing is, \S\ref{sec:eval}), and thus we opt to evaluate \sys{} without ride sharing. We note that the relative consensus overheads decrease linearly as $n$ scales, and are thus already exceedingly small for $n >> k$ (max \#instances).

Nonetheless, for completeness, we discuss how to theoretically implement ride-sharing in \sys{}. 


\par \textbf{Shared Cars.} Consensus messages can be embedded on the lane of the respective instance leader (Fig \ref{fig:ride-share}); other lanes remain unaffected. Prepare and Confirm phases are piggybacked on a car each while \textsc{Timeouts} may ride any message. Multiple consensus instances ({\em i.e.,} messages for different slots and views) may hitch a ride on the same car and share one "batched" signature. In exchange, a replica may need to forward all messages in a batch to prove the validity of a single signed message; one can reduce  costs by signing a merkle root of the batch, attaching a small merkle path to authenticate individual messages, and caching repeat signatures~\cite{suri2021basil}.

\begin{figure}[h!]
    \vskip 0pt
   \centering
   \includegraphics[width=0.45\textwidth]{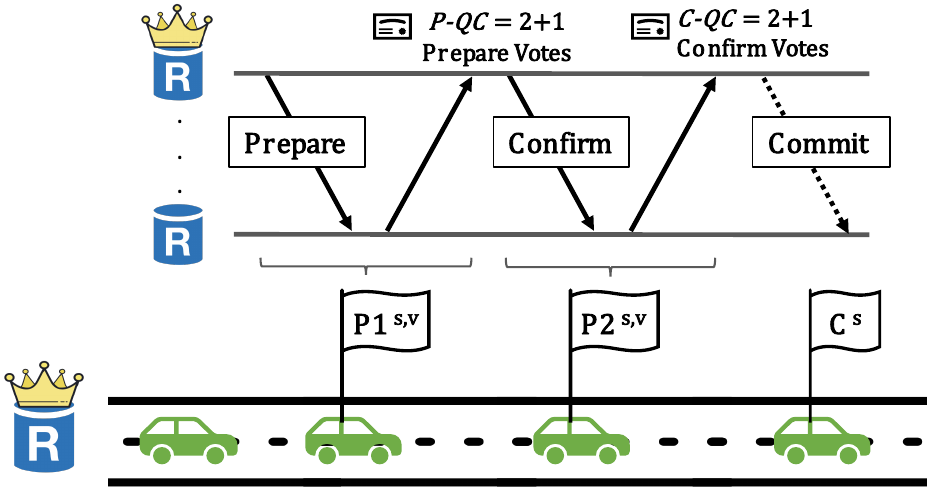}
   \caption{Core Consensus coordination pattern. Messages may be piggybacked on existing cars: P1 = \textsc{Prepare}, P2 = \textsc{Confirm}, C = \textsc{Commit}.}
   \label{fig:ride-share}
   \vskip 0pt
\end{figure}

Cars, by default, only wait for $f+1$ votes -- too few to successfully assemble PrepareQCs and CommitQCs. To account for this we upgrade any car carrying one or more \textsc{Prepare/Confirm} consensus messages into a special \textit{ambassador} car that waits for $n-f$ votes; ambassadors representing \textsc{Prepares} may attempt to briefly wait for the full $n$ \textsc{Prepare-Votes} required to construct a Fast-CommitQC. Fig. \ref{fig:consensus} illustrates the embedding, flags denote ambassadors.
Correct replicas will always vote to confirm \textsc{Proposal} availability, but may omit a consensus-vote if they have changed views for the respective slot. To distinguish faulty voters from correct replicas \sys{} delays cars following an ambassador: ambassador cars wait for the earlier of QC formation, and expiration of a short timer (beyond the first $f+1$ replies (PoA)). Upon timing out, a lane proceeds with the next car, and continues consensus processing for late arrivals individually or on a later car. We remark that although the data proposal pace may decelerate due to increased wait-time for ambassador cars, lane growth fundamentally remains asynchronous. 

\par \textbf{Ridin' Solo.}
It may not always be convenient for consensus messages to hitch a ride. In the worst-case, a consensus message that is available just after starting the current car may need to wait up to 2md until the next \textsc{Proposal} (1md in expectation). Ride sharing thus poses a trade-off between optimal resource utilization and latency. To strike a middle-ground, consensus message may wait for a small timeout to see whether they can hitch a ride, and otherwise drive themselves.
However, rather than sending consensus messages fully independently, \sys{} may once again piggyback them alongside data proposals in a special \textit{off-lane car} to maximize message utility. 

\fs{the below is actually more subtle than I describe, so it may be worth to just describe as sending externally. Detail in comment:}

This off-lane car follows an ambassador car's usual structure, and may carry both multiple consensus instance messages. Uniquely, however, it does not have a lane position, and does not obey lane rules (certified chained FIFO).

In order to minimize end-to-end latency, the off-lane car includes as parent reference the latest proposal (yet to be certified) proposal sent by the leader. This is safe and live as it is equivalent to an optimistic leader tip; the consensus leader is responsible for the delivery of it's tip history. 

Since the off-lane car lacks a position, normal cars in the leaders lane do not extend the off-lane car. To avoid orphaning an off-lane car whose payload was certified but did not commit, a subsequent car may be augmented to include a special (optional) \textit{tow} reference to a certificate of an off-lane car. This allows the off-lane car to eventually be committed as part of the towing car's history.

\end{document}